\documentclass{amsart}
\usepackage{hyperref}
\usepackage{graphicx}
\usepackage{curves}

\newcommand{\arxiv}[1]{\href{http://arxiv.org/abs/#1}{arXiv:#1}}
\newcommand*{\mailto}[1]{\href{mailto:#1}{\nolinkurl{#1}}}

\newtheorem{theorem}{Theorem}[section]
\newtheorem{lemma}[theorem]{Lemma}
\newtheorem{corollary}[theorem]{Corollary}
\newtheorem{remark}[theorem]{Remark}
\newtheorem{hypothesis}[theorem]{Hypothesis}

\newcommand{\R}{{\mathbb R}}
\newcommand{\N}{{\mathbb N}}
\newcommand{\Z}{{\mathbb Z}}
\newcommand{\C}{{\mathbb C}}
\newcommand{\T}{{\mathbb T}}

\newcommand{\nn}{\nonumber}
\newcommand{\be}{\begin{equation}}
\newcommand{\ee}{\end{equation}}

\newcommand{\ol}{\overline}
\newcommand{\ti}{\tilde}
\newcommand{\wti}{\widetilde}

\newcommand{\abs}[1]{\lvert#1 \rvert}

\newcommand{\id}{\mathbb{I}}
\newcommand{\I}{\mathrm{i}}
\newcommand{\E}{\mathrm{e}}
\newcommand{\tr}{\operatorname{tr}}
\newcommand{\ind}{\operatorname{ind}}
\newcommand{\re}{\operatorname{Re}}
\newcommand{\im}{\operatorname{Im}}
\DeclareMathOperator{\res}{Res}

\newcommand{\lz}{\ell^2(\Z)}

\def\Xint#1{\mathchoice
   {\XXint\displaystyle\textstyle{#1}}%
   {\XXint\textstyle\scriptstyle{#1}}%
   {\XXint\scriptstyle\scriptscriptstyle{#1}}%
   {\XXint\scriptscriptstyle\scriptscriptstyle{#1}}%
   \!\int}
\def\XXint#1#2#3{{\setbox0=\hbox{$#1{#2#3}{\int}$}
     \vcenter{\hbox{$#2#3$}}\kern-.5\wd0}}
\def\dashint{\Xint-}

\newcommand{\eps}{\varepsilon}

\newcommand{\sig}{\sigma}
\newcommand{\lam}{\lambda}
\newcommand{\gam}{\gamma}

\numberwithin{equation}{section}

\newcommand{\sigI}{\begin{pmatrix} 0 & 1 \\ 1 & 0 \end{pmatrix}}
\newcommand{\ssigI}{\big(\begin{smallmatrix} 0 & 1 \\ 1 & 0\end{smallmatrix}\big)}
\newcommand{\rI}{\begin{pmatrix}  1 & 1 \end{pmatrix}}

\begin{document}

\title[Long-Time Asymptotics of the Toda Lattice]{Long-Time Asymptotics of the Toda Lattice for Decaying Initial Data Revisited}

\author[H. Kr\"uger]{Helge Kr\"uger}
\address{Department of Mathematics\\ Rice University\\ Houston\\ TX 77005\\ USA}
\email{\mailto{helge.krueger@rice.edu}}
\urladdr{\url{http://math.rice.edu/~hk7/}}

\author[G. Teschl]{Gerald Teschl}
\address{Faculty of Mathematics\\
Nordbergstrasse 15\\ 1090 Wien\\ Austria\\ and International Erwin Schr\"odinger
Institute for Mathematical Physics, Boltzmanngasse 9\\ 1090 Wien\\ Austria}
\email{\mailto{Gerald.Teschl@univie.ac.at}}
\urladdr{\url{http://www.mat.univie.ac.at/~gerald/}}

\thanks{Research supported by the Austrian Science Fund (FWF) under Grant No.\ Y330.}
\thanks{Rev. Math. Phys. {\bf 21:1}, 61--109 (2009)}


\begin{abstract}
The purpose of this article is to give a streamlined and self-contained
treatment of the long-time asymptotics of the
Toda lattice for decaying initial data in the soliton and in the similarity region
via the method of nonlinear steepest descent.
\end{abstract}

\maketitle

\section{Introduction}

The simplest model of a solid is a chain of particles with nearest neighbor interaction.
The Hamiltonian of such a system is given by
\begin{equation}
\mathcal{H}(p,q) = \sum_{n\in\Z} \Big(\frac{p(n,t)^2}{2} + V(q(n+1,t) - q(n,t)) \Big),
\end{equation}
where $q(n,t)$ is the displacement of the $n$-th particle from its equilibrium position,
$p(n,t)$ is its momentum (mass $m=1$), and $V(r)$ is the interaction potential.

Restricting the attention to finitely many particles (e.g., by imposing periodic boundary
conditions) and to the harmonic interaction $V(r)=\frac{r^2}{2}$, the equations of motion form a
linear system of differential equations with constant coefficients. The solution
is then given by a superposition of the associated {\em normal modes}. Around 1950
it was generally believed that a generic nonlinear perturbation would yield to {\em thermalization}.
That is, for any initial condition the energy should eventually be equally distributed
over all normal modes.  In 1955 Enrico Fermi, John Pasta, and Stanislaw Ulam carried out a seemingly
innocent computer experiment at Los Alamos, \cite{fpu}, to investigate the rate of approach to the
equipartition of energy.  However, much to everybody's surprise,
the experiment indicated, instead of the expected thermalization, a quasi-periodic motion
of the system! Many attempts were made to explain this result but it was not until ten years later
that Martin Kruskal and Norman Zabusky, \cite{zakr}, revealed the connections with solitons
(see \cite{dpr} for further historical information and a pedagogical discussion).

This had a big impact on soliton mathematics and led to an explosive growth in the last
decades. In particular, it led to the search for a potential $V(r)$ for which the above
system has soliton solutions. By considering addition formulas for elliptic functions,
Morikazu Toda came up with the choice $V(r) = {\rm e}^{- r} + r -1$.
The corresponding system is now known as the Toda equation, \cite{ta}.

The equation of motion in this case reads explicitly
\begin{align}\nn
\frac{d}{dt} p(n,t) &= -\frac{\partial\mathcal{H}(p,q)}{\partial q(n,t)}
= {\rm e}^{-(q(n,t) - q(n-1,t))} - {\rm e}^{-(q(n+1,t) - q(n,t))},\\
\frac{d}{dt} q(n,t) &= \frac{\partial\mathcal{H}(p,q)}{\partial p(n,t)} = p(n,t).
\end{align}

The important property of the Toda equation is the existence of so called soliton solutions,
that is, pulslike waves  which spread in time without changing their size or shape and interact with
each other in a particle-like way.
This is a surprising phenomenon, since for a generic linear equation one would expect spreading
of waves (dispersion) and for a generic nonlinear force one
would expect that solutions only exist for a finite time (breaking of waves). Obviously our
particular force is such that both phenomena cancel each other giving rise to a stable
wave existing for all time!

In fact, in the simplest case of one soliton, you can easily verify that this solution is given by
\begin{equation}
q_1(n,t) = q_+ + \log\left(\frac{1 + \frac{\gam}{1-\E^{-2\kappa}} \exp(-2\kappa n + 2\sig\sinh(\kappa)t)}{1
+ \frac{\gam}{1-\E^{-2\kappa}} \exp(-2\kappa (n+1) + 2\sig\sinh(\kappa)t)}\right),
\end{equation}
with $\kappa,\gam>0$ and $\sig\in\{\pm1\}$.
\begin{figure}[th]
\centering
\includegraphics[width=6cm]{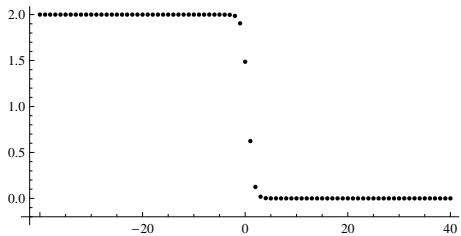}
\caption{One soliton $q_1(n,0)$ with $\kappa=1$, $\gam=1$, and $q_0=0$.}
\end{figure}
It describes a single bump traveling through the crystal with speed
$\sig\sinh(\kappa)/ \kappa$ and width proportional to $1/ \kappa$.
In other words, the smaller the soliton the faster it propagates. It results in
a total displacement $2 \kappa$ of the crystal.

However, this is just the tip of the iceberg and can be generalized to
the $N$-soliton solution
\begin{equation}
q_N(n,t) = q_+ + \log\left(\frac{\det(\id + C_N(n,t))}{\det(\id + C_N(n+1,t))}\right),
\end{equation}
where
\begin{equation}
C_N(n,t) = \left(\frac{\sqrt{\gam_i(n,t) \gam_j(n,t)}}{1-
\mathrm{e}^{-(\kappa_i+\kappa_j)}} \right)_{1\le i,j\le N}, \quad
\gam_j(n,t)= \gam_j \mathrm{e}^{ -2\kappa_j n - 2\sig_j\sinh(\kappa_j) t},
\end{equation}
with $\kappa_j,\gam_j>0$ and $\sig_j\in\{\pm1\}$. The case $N=1$ coincides with
the one soliton solution from above and asymptotically, as $t\to\infty$, the $N$-soliton
solution can be written as a sum of one-soliton solutions.

Historically such {\em solitary waves} were first observed by the naval architect John
Scott Russel \cite{rus}, who followed the bow wave of a barge which moved along a channel
maintaining its speed and size (see the review article \cite{pal} for further
information).

The importance of these solitary waves is that they constitute the stable part of the solutions arising from
arbitrary short range initial conditions and can be used to explain the quasi-periodic
behaviour found by Fermi, Pasta, and Ulam. In fact, the classical result
discovered by Zabusky and Kruskal \cite{zakr} states that every "short range" initial condition
eventually splits into a number of stable solitons and a decaying background radiation component.
\begin{figure}
\includegraphics[width=8cm]{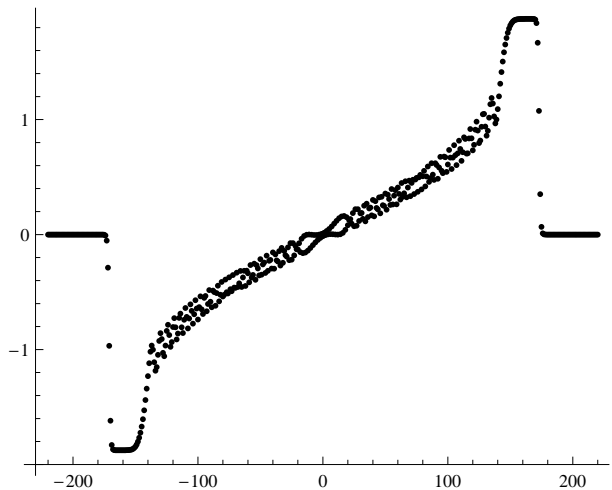}
\caption{Numerically computed solution $q(n,150)$ of the Toda lattice, with initial
condition all particles at rest in their equilibrium positions except for the one in the middle which is
displaced by $1$.} \label{fig1}
\end{figure}
This is illustrated in Figure~\ref{fig1} which shows the numerically computed
solution $q(n,t)$ corresponding to the initial condition $q(n,0)= \delta_{0,n}$, $p(n,0)=
0$ at some large time $t=130$. You can see the soliton region $|\frac{n}{t}|>1$
with two single soliton on the very left respectively right and the similarity region $|\frac{n}{t}|<1$ where
there is a continuous displacement plus some small oscillations which decay like $t^{-1/2}$ and
are asymptotically given by
\begin{equation}
q(n,t) \asymp 2\log(T_0(z_0)) + \left(\frac{2 \nu(z_0)}{-\sin(\theta_0) t}\right)^{1/2} \!
\cos\big(t \Phi_0(z_0) + \nu(z_0)\log(t) - \delta(z_0)\big),
\end{equation}
where $z_0 =\E^{\I\theta_0}$ is a {\em slow variable} depending only on $\frac{n}{t}$ and the functions $T_0(z_0)$,
$\nu(z_0)$, $\Phi_0(z_0)$, and $\delta(z_0)$ are explicitly given in terms of the scattering data associated
with the initial data. Our main goal will be to mathematically justify this formula for the solution
in  the similarity region $|\frac{n}{t}|<1$ (Theorem~\ref{thm:asym2}) and to show that the solution splits into a number of
solitons in the soliton region $|\frac{n}{t}|>1$ (Theorem~\ref{thm:asym}).

Existence of soliton solutions is usually connected to complete integrability of
the system, and this is also true for the Toda equation.
To see that the Toda equation is indeed integrable we introduce Flaschka's variables \cite{fl1}
\begin{equation}
a(n,t) = \frac{1}{2} {\rm e}^{-(q(n+1,t) - q(n,t))/2}, \qquad
b(n,t) = -\frac{1}{2} p(n,t)
\end{equation}
and obtain the form most convenient for us
\begin{align} \nn
\frac{d}{dt} a(t) &= a(t) \Big(b^+(t)-b(t)\Big), \\ \label{todeqfl}
\frac{d}{dt} b(t) &= 2 \Big(a(t)^2-a^-(t)^2\Big).
\end{align}
Here we have used the abbreviation
\begin{equation}
f^\pm(n)= f(n\pm1).
\end{equation}

Note that if $q(n,t)\to q_\pm$ sufficiently fast as $n\to\pm\infty$, the converse map is given by
\begin{equation}
q(n,t) = q_+ + 2\log\left( \prod_{j=n}^\infty (2 a(j,t))\right), \qquad
p(n,t) = -2 b(n,t).
\end{equation}
Moreover, $q(n,t)\to q_\pm$, $p(n,t)\to 0$ as $|n|\to\infty$ corresponds to $a(n,t)\to \frac{1}{2}$, $b(n,t)\to 0$.

To show complete integrability it suffices to find a so-called Lax pair \cite{lax}, that
is, two operators $H(t)$, $P(t)$ in $\lz$ such that the Lax equation
\begin{equation} \label{laxeq}
\frac{d}{dt} H(t) = P(t) H(t) - H(t) P(t)
\end{equation}
is equivalent to \eqref{todeqfl}. One can easily convince oneself that the right
choice is
\begin{align} \nn
H(t) &= a(t) S^+ + a^-(t) S^- + b(t),\\
P(t) &= a(t) S^+ - a^-(t) S^-,
\end{align}
where $(S^\pm f)(n) = f^\pm(n)= f(n\pm1)$ are the shift operators.
Now the Lax equation \eqref{laxeq} implies that the operators $H(t)$ for
different $t\in\R$ are unitarily equivalent (cf.\ \cite[Thm.~12.4]{tjac}):

\begin{theorem}\label{thmunitary}
Let $P(t)$ be a family of bounded skew-adjoint operators, such that $t\mapsto
P(t)$ is differentiable. Then there exists a family of unitary propagators $U(t,s)$
for $P(t)$, that is,
\begin{equation}
\frac{d}{dt} U(t,s) = P(t) U(t,s), \qquad U(s,s)=\id.
\end{equation}
Moreover, the Lax equation \eqref{laxeq} implies
\begin{equation}
H(t)= U(t,s) H(s) U(t,s)^{-1}.
\end{equation}
\end{theorem} 

This result has several important consequences. First of all it implies global existence
of solutions of the Toda lattice. In fact, considering the Banach space of
all bounded real-valued coefficients $(a(n),b(n))$ (with the sup norm), local existence
follows from standard results for differential equations in Banach spaces. Moreover,
Theorem~\ref{thmunitary} implies that the norm $\|H(t)\|$ is constant, which in turn
provides a uniform bound on the coefficients of $H(t)$,
\begin{equation}
\|a(t)\|_\infty + \|b(t)\|_\infty \le 2\|H(t)\| = 2\|H(0)\|.
\end{equation}
Hence solutions of the Toda lattice cannot blow up and are global in time (see \cite[Sect.~12.2]{tjac}
for details).

Second, it provides an infinite sequence of conservation laws expected from a completely
integrable system. Indeed, if the Lax equation \eqref{laxeq} holds for $H(t)$, it automatically also holds
for $H(t)^j$. Taking traces shows that
\begin{equation}
\tr\big( H(t)^j - H_0^j\big), \quad j\in\N,
\end{equation}
is an infinite sequence of conserved quantities, where $H_0$ is the operator corresponding to
the constant solution $a_0(n,t) = \frac{1}{2}$, $b_0(n,t)=0$ (it is
needed to make the trace converge). Introducing a suitable symplectic structure, they can be shown
to be in involution as well (\cite[Sect.~1.7]{ghmt}). For example,
\begin{align} \nn
\tr\big( H(t) - H_0\big) &= \sum_{n\in\Z} b(n,t) = -\frac{1}{2} \sum_{n\in\Z}
p(n,t) \mbox{ and}\\
\tr\big( H(t)^2 - H_0^2\big) &= \sum_{n\in\Z} b(n,t)^2 + 2(a(n,t)^2 - \frac{1}{4}) =
\frac{1}{2} \mathcal{H}(p,q)
\end{align}
correspond to conservation of the total momentum and the total energy,
respectively.

These observations pave the way for a solution of the Toda equation via
the inverse scattering transform originally invented by Gardner, Green, Kruskal, and Miura  \cite{ggkm}
for the Korteweg--De Vries equation (see \cite[Sect.~13.4]{tjac} for the case of the Toda lattice).
In particular, Theorem~\ref{thmunitary} implies that the operators $H(t)$, $t\in\R$,
are unitarily equivalent and that the spectrum $\sig(H(t))$ is independent of $t$. Now the general idea is to find
suitable spectral data $S(H(t))$ for $H(t)$ which uniquely determine $H(t)$. Then equation
\eqref{laxeq} can be used to derive linear evolution equations for $S(H(t))$ which are easy to solve.
In our case these data will be the so called scattering data and the formal procedure (which can be
thought of as a nonlinear Fourier transform) is summarized below:

\hspace*{3mm}\begin{picture}(12,4.5)
\put(1.25,0){$(a(0),b(0))$}
\put(7.5,0){$(a(t),b(t))$}
\put(2.1,0.5){\vector(0,1){2}}
\put(0.2,1.5){direct}
\put(0.2,1.1){scattering}
\put(8.3,2.5){\vector(0,-1){2}}
\put(8.5,1.5){inverse}
\put(8.5,1.1){scattering}
\put(1.3,3){$S(H(0))$}
\put(7.5,3){$S(H(t))$}
\put(3.5,3.1){\vector(1,0){3.4}}
\put(4,3.3){time evolution}
\end{picture}\\[2mm]
The inverse scattering step will be done by reformulating the problem as a Riemann--Hilbert
factorization problem. This Riemann--Hilbert problem will then be analyzed using the method
of nonlinear steepest descent by Deift and Zhou \cite{dz} (which is the nonlinear analog of the
steepest descent for Fourier type integrals). In fact, one of our goals is to give a complete and
expository introduction to this method.  We are trying to present a streamlined and
simplified approach with complete proofs. In particular, we have added two appendices
which show how to solve the localized Riemann--Hilbert problem on a  small cross via
parabolic cylinder functions and how to rewrite Riemann--Hilbert problems as singular
integral equations. Only some basic knowledge on Riemann--Hilbert problems, which
can be found for example in the beautiful lecture notes by Deift \cite{deiftbook}, is required.

For further information on the history of the steepest descent method, which was inspired by
earlier work of Manakov \cite{ma} and Its \cite{its}, and the problem of finding the long-time
asymptotics for integrable nonlinear wave equations, we refer to the survey by Deift, Its, and
Zhou \cite{diz}.

More information on the Toda lattice can be found in the monographs by Faddeev and Takhtajan \cite{fad},
Gesztesy, Holden, Michor, and Teschl \cite{ghmt}, Teschl \cite{tjac}, or Toda \cite{ta}. Here we partly
followed the review article \cite{taet}. A much more comprehensive guide to the literature can be found in
Section~1.8 of \cite{ghmt}.

First results on the long-time asymptotics of the doubly infinite Toda lattice were given by
Novokshenov and Habibullin \cite{nh} and Kamvissis \cite{km}. Long-time asymptotics for
the finite and semi infinite Toda lattice can be found in Moser \cite{mo} and Deift, Li, and Tomei \cite{dlt}, respectively.
The long-time behaviour of Toda shock problem was investigated by Kamvissis \cite{km2} and Venakides, Deift, and Oba
\cite{vdo} and of the Toda rarefaction problem by Deift, Kamvissis, Kriecherbauer, and Zhou \cite{dkkz}. For the case of
a periodic driving force see Deift, Kriecherbauer, and Venakides \cite{dkv}.

Finally, we also want to mention that one could replace the constant background solution
by a periodic one. However, this case exhibits a much different behaviour, as was pointed
out by Kamvissis and Teschl in \cite{kt} (see also \cite{emtist}, \cite{emtsr}, \cite{kt2}, \cite{kt3},
and \cite{krt2} for a rigorous mathematical treatment).

\section{Main results}

As stated in the introduction, we want to compute the long-time asymptotics for
the doubly infinite Toda lattice which reads in Flaschka's
variables
\be \label{tl}
\aligned
\dot b(n,t) &= 2(a(n,t)^2 -a(n-1,t)^2),\\
\dot a(n,t) &= a(n,t) (b(n+1,t) -b(n,t)),
\endaligned
\ee
$(n,t) \in \Z \times \R$. Here the dot denotes differentiation with
respect to time. We will consider solutions $(a,b)$ satisfying
\be \label{decay}
\sum_n (1+|n|)^{l+1} (|a(n,t) - \frac{1}{2}| + |b(n,t)|) < \infty
\ee
for some $l\in\N$ for one (and hence for all, see \cite{tjac}) $t\in\R$. It is
well-known that the corresponding initial value problem has unique global solutions
which can be computed via the inverse scattering transform \cite{tjac}.

The long-time asymptotics were first derived by Novokshenov and Habibullin \cite{nh} and
were later made rigorous by Kamvissis \cite{km} under the additional assumption that no
solitons are present. The case of solitons was recently investigated by us in \cite{krt}.

As one of our main simplifications in contradistinction to \cite{km} we will work with the vector
Riemann--Hilbert problem which arises naturally from the inverse scattering theory, thus avoiding the detour
over the associated matrix Riemann--Hilbert problem. This also avoids the singularities appearing
in the matrix Riemann--Hilbert problem in case the reflection coefficient is $-1$ at the
band edges.

To state the main results, we begin by recalling that the sequences $a(n,t)$, $b(n,t)$, $n\in\Z$,
for fixed $t\in\R$, are uniquely determined by its scattering data, that is, by its right reflection
coefficient $R_+(z,t)$, $|z|=1$, and its eigenvalues $\lam_j\in(-\infty,-1)\cup(1,\infty)$, $j=1,\dots, N$,
together with the corresponding right norming constants $\gam_{+,j}(t)>0$, $j=1,\dots, N$.
It is well-known that under the assumption \eqref{decay} the reflection coefficients are $C^{l+1}(\T)$.
Rather than in the complex plane, we will work on the unit disc using the usual
Joukowski transformation
\be\label{defzlam}
\lam = \frac{1}{2} \left(z + \frac{1}{z}\right),\quad z= \lam - \sqrt{\lam^2 -1}, \qquad
\lam\in\C, \: |z|\leq 1.
\ee
In these new coordinates the eigenvalues $\lam_j\in(-\infty,-1)\cup(1,\infty)$ will be denoted by
$\zeta_j\in(-1,0)\cup(0,1)$. The continuous spectrum $[-1,1]$ is mapped to the unit circle
$\T$. Moreover, the phase of the associated Riemann--Hilbert problem is given by
\begin{equation} \label{eq:Phi}
\Phi(z)=z-z^{-1}+2 \frac{n}{t} \log(z)
\end{equation}
and the stationary phase points, $\Phi'(z)=0$, are denoted by
\be
z_0=  -\frac{n}{t} - \sqrt{(\frac{n}{t})^2 -1}, \quad z_0^{-1}= -\frac{n}{t} + \sqrt{(\frac{n}{t})^2 -1}
\ee
and correspond to
\be
\lam_0=-\frac{n}{t}.
\ee
Here the branch of the square root is chosen such that $\im(\sqrt{z})\ge 0$.
For $\frac{n}{t}<-1$ we have $z_0\in(0,1)$, for $-1\le \frac{n}{t} \le1$
we have $z_0\in\T$ (and hence $z_0^{-1}=\ol{z_0}$), and for $\frac{n}{t}>1$
we have $z_0\in(-1,0)$. For $|\frac{n}{t}|>1$ we will also need the value
$\zeta_0\in(-1,0)\cup(0,1)$ defined via $\re(\Phi(\zeta_0))=0$, that is,
\be
\frac{n}{t} = -\frac{\zeta_0 - \zeta_0^{-1}}{2\log(|\zeta_0|)}.
\ee
We will set $\zeta_0=-1$ if $|\frac{n}{t}|\le 1$ for notational convenience.
A simple analysis shows that for $\frac{n}{t}<-1$ we have $0<\zeta_0 < z_0 <1$ and
for $\frac{n}{t}>1$ we have $-1<z_0 < \zeta_0 <0$.

Furthermore, recall that the transmission coefficient $T(z)$, $|z|\le 1$, is time independent
and can be reconstructed using the Poisson--Jensen formula. In particular, we define
the partial transmission coefficient with respect to $z_0$ by
\begin{align}\nn
T(z,z_0) &=\\ \label{def:Tzz0}
& \begin{cases}
\prod\limits_{\zeta_k\in(\zeta_0,0)} |\zeta_k| \frac{z-\zeta_k^{-1}}{z-\zeta_k}, & z_0 \in (-1,0), \\
\left(\prod\limits_{\zeta_k\in(-1,0)} |\zeta_k| \frac{z-\zeta_k^{-1}}{z-\zeta_k} \right)
\exp\left(\frac{1}{2\pi\I}\int\limits_{\ol{z_0}}^{z_0}\log(|T(s)|) \frac{s+z}{s-z} \frac{ds}{s}\right), & |z_0| = 1, \\
\left(\prod\limits_{\zeta_k\in(-1,0)\cup(\zeta_0,1)}\!\!\! |\zeta_k| \frac{z-\zeta_k^{-1}}{z-\zeta_k} \right)
\exp\left(\frac{1}{2\pi\I}\int\limits_{\T}\log(|T(s)|) \frac{s+z}{s-z} \frac{ds}{s}\right), & z_0 \in (0,1).
\end{cases}
\end{align}
Here, in the case $z_0\in\T$,  the integral is to be taken along the arc $\Sigma(z_0)= \{z \in\T | \re(z)<\re(z_0)\}$
oriented counterclockwise. For $z_0\in(-1,0)$ we set $\Sigma(z_0)=\emptyset$ and for $z_0\in(0,1)$ we
set $\Sigma(z_0)=\T$. Then $T(z,z_0)$ is meromorphic for $z\in\C\backslash\Sigma(z_0)$.
Observe that $T(z,z_0)=T(z)$ once $z_0\in(0,1)$ and $(0,\zeta_0)$ contains no eigenvalues. Moreover,
$T(z,z_0)$ can be computed in terms of the scattering data since $|T(z)|^2= 1- |R_+(z,t)|^2= 1- |R_+(z,0)|^2$.

Moreover, we set
\begin{align}\nn
T_0(z_0) &= T(0,z_0)\\
&= \begin{cases}
\prod\limits_{\zeta_k\in(\zeta_0,0)} |\zeta_k|^{-1}, & z_0 \in (-1,0),\\
\left(\prod\limits_{\zeta_k\in(-1,0)} |\zeta_k|^{-1} \right)
\exp\left(\frac{1}{2\pi\I}\int\limits_{\ol{z_0}}^{z_0}\log(|T(s)|) \frac{ds}{s}\right), & |z_0| = 1, \\
\left(\prod\limits_{\zeta_k\in(-1,0)\cup(\zeta_0,1)} |\zeta_k|^{-1} \right)
\exp\left(\frac{1}{2\pi \I}\int\limits_{\T}\log(|T(s)|) \frac{ds}{s}\right), & z_0 \in (0,1),
\end{cases}
\end{align}
and
\begin{align} \nn
T_1(z_0) &= \frac{\partial}{\partial z} \log T(z,z_0) \Big|_{z=0}\\
&= \begin{cases}
\sum\limits_{\zeta_k\in(\zeta_0,0)} (\zeta_k^{-1} -\zeta_k), & z_0 \in (-1,0),\\
\sum\limits_{\zeta_k\in(-1,0)} (\zeta_k^{-1} -\zeta_k) +
\frac{1}{\pi\I}\int\limits_{\ol{z_0}}^{z_0}\log(|T(s)|) \frac{ds}{s^2}, & |z_0| = 1, \\
\sum\limits_{\zeta_k\in(-1,0)\cup(\zeta_0,1)} (\zeta_k^{-1} -\zeta_k) +
\frac{1}{\pi \I}\int\limits_{\T}\log(|T(s)|) \frac{ds}{s^2}, & z_0 \in (0,1).
\end{cases}
\end{align}
In other words, $T(z,z_0)=T_0(z_0) ( 1 + T_1(z_0) z + O(z^2))$.

\begin{theorem}[Soliton region]\label{thm:asym}
Assume \eqref{decay} for some $l\in\N$ and abbreviate by $c_k= -\frac{\zeta_k - \zeta_k^{-1}}{2\log(|\zeta_k|)}$
the velocity of the $k$'th soliton determined by $\re(\Phi(\zeta_k))=0$.
Then the asymptotics in the soliton region, $|n/t| \geq 1 + C/t \log(t)^2$ for any
$C>0$, are as follows.

Let $\eps > 0$ sufficiently small such that the intervals
$[c_k-\eps,c_k+\eps]$, $1\le k \le N$, are disjoint and lie inside $(-\infty,-1)\cup(1,\infty)$.

If $|\frac{n}{t} - c_k|<\eps$ for some $k$, the solution is asymptotically given by a single soliton
\begin{align}\nn
\prod_{j=n}^\infty (2 a(j,t)) &= T_0(z_0) \left(
\sqrt{\frac{1-\zeta_k^2 + \gam_k(n,t)}{1-\zeta_k^2 + \gam_k(n+1,t)}} + O(t^{-l}) \right),\\ \label{eqsola1}
\sum_{j=n+1}^\infty b(j,t) &= \frac{1}{2} T_1(z_0) +
\frac{\gam_k(n,t)  \zeta_k (1-\zeta_k^2)}{2((\gam_k(n,t) -1) \zeta_k^2+1)} + O(t^{-l}),
\end{align}
where
\be
\gam_k(n,t) = \gam_k T(\zeta_k,-c_k - \sqrt{c_k^2 -1})^{-2} \E^{t (\zeta_k - \zeta_k^{-1})} \zeta_k^{2n}.
\ee

If $|\frac{n}{t} -c_k| \geq \eps$, for all $k$, one has
\begin{align}\nn
\prod_{j=n}^\infty (2 a(j,t)) &= T_0(z_0) \left(1 + O(t^{-l}) \right),\\ \label{eqsola2}
\sum_{j=n+1}^\infty b(j,t) &= \frac{1}{2} T_1(z_0) + O(t^{-l}).
\end{align}
\end{theorem}

Note that one can choose $|\frac{n}{t} - c_k|<\eps_1$ for the regions where \eqref{eqsola1} is valid, respectively
$|\frac{n}{t} -c_k| \geq \eps_2$ for the regions where \eqref{eqsola2} is valid, such that the regions overlap if $\eps_1>\eps_2$.
Due to the exponential decay of the one-soliton solution, both formulas of course produce the same result on the overlap.

In particular, we recover the well-known fact that the solution splits into a sum of independent solitons
where the presence of the other solitons and the radiation part corresponding to the continuous spectrum
manifests itself in phase shifts given by $T(\zeta_k,-c_k - \sqrt{c_k^2 -1})^{-2}$. Indeed, notice that
for $\zeta_k\in (-1,0)$ this term just contains the product over the Blaschke factors corresponding to solitons
$\zeta_j$ with $\zeta_k<\zeta_j$. For $\zeta_k\in (0,1)$ we have the product over the Blaschke factors
corresponding to solitons $\zeta_j\in(-1,0)$, the integral over the full unit circle, plus the product over the
Blaschke factors corresponding to solitons $\zeta_j$ with $\zeta_k>\zeta_j$.

Furthermore, this result shows that in the region $\frac{n}{t}>1$ the solution is asymptotically given by a $N_-$-soliton solution,
where $N_-$ is the number of $\zeta_j\in(-1,0)$, formed from the data $\zeta_j$, $\gam_j$ for
all $\zeta_k\in(-1,0)$. Similarly, in the region $\frac{n}{t}<-1$ the solution is asymptotically given by a $N_+$-soliton solution,
where $N_+$ is the number of $\zeta_j\in(0,1)$, formed from the data $\zeta_j$, $\ti{\gam}_j$ for
all $\zeta_j\in(0,1)$, where
\be
\ti{\gam}_j = \gam_j \left(\prod\limits_{\zeta_k\in(-1,0)} |\zeta_k| \frac{\zeta_j-\zeta_k^{-1}}{\zeta_j-\zeta_k} \right)
\exp\left(\frac{1}{2\pi\I}\int\limits_{\T}\log(|T(s)|) \frac{s+\zeta_j}{s-\zeta_j} \frac{ds}{s}\right).
\ee

In the remaining region, we will show

\begin{theorem}[Similarity region]\label{thm:asym2}
Assume \eqref{decay} with $l\ge 5$, then, away from the soliton region, $|n/t| \leq 1 - C$ for any
$C>0$, the asymptotics are given by
\begin{align}\nn
\prod_{j=n}^\infty (2 a(j,t)) = & T_0(z_0) \Bigg(1 + \left(\frac{\nu(z_0)}{- 2\sin(\theta_0) t}\right)^{1/2}
\cos\big(t \Phi_0(z_0) + \nu(z_0) \log(t) - \delta(z_0)\big)\\
& + O(t^{-\alpha}) \Bigg),\\ \nn
\sum_{j=n+1}^\infty b(j,t) = &  \frac{1}{2} T_1(z_0) + \left(\frac{\nu(z_0)}{-2\sin(\theta_0) t}\right)^{1/2}
\cos\big(t \Phi_0(z_0) + \nu(z_0) \log(t) - \delta(z_0) + \theta_0\big)\\
& +O(t^{-\alpha}), \qquad z_0 = \E^{\I\theta_0},
\end{align}
for any $\alpha<1$. Here
\begin{align}\nn
\nu(z_0) =& -\frac{1}{\pi} \log(|T(z_0)|),\\ \nn
\Phi_0(z_0) =& 2 (\sin(\theta_0) - \theta_0\cos(\theta_0)), \\\label{eq:defvar}
\delta(z_0) =& \pi/4 - 3 \nu(z_0) \log |2\sin(\theta_0)| + 2 \arg(\ti{T}(z_0)) - \arg(R_+(z_0,0))\\ \nn
& +\arg(\Gamma(\I\nu(z_0))), \\ \nn
\ti{T}(z_0) = &\prod\limits_{\zeta_k\in(-1,0)} |\zeta_k| \frac{z-\zeta_k^{-1}}{z-\zeta_k} \cdot
\exp\left(\frac{1}{2\pi\I}\int\limits_{\ol{z_0}}^{z_0}\log\Big(\frac{|T(s)|}{|T(z_0)|}\Big) \frac{s+z_0}{s-z_0} \frac{ds}{s}\right),
\end{align}
and $\Gamma(z)$ is the gamma function.
\end{theorem}

For $a(n,t)$ respectively $b(n,t)$ we obtain as a simple consequence:

\begin{corollary}
Under the same assumptions as in Theorem~\ref{thm:asym2} we have
\begin{align}\nn
a(n,t) = & \frac{1}{2}  + \left(\frac{-\sin(\theta_0) \nu(z_0)}{2 t}\right)^{1/2}
\cos\big(t \Phi_0(z_0) + \nu(z_0) \log(t) - \delta(z_0)-\theta_0\big)\\
& +O(t^{-\alpha}),\\ \nn
b(n,t) = & \left(\frac{-2\sin(\theta_0) \nu(z_0)}{t}\right)^{1/2} 
\sin\big(t \Phi_0(z_0) + \nu(z_0) \log(t) - \delta(z_0) + 2\theta_0\big)\\
& +O(t^{-\alpha}).
\end{align}
\end{corollary}

\begin{proof}
To get the first formula for we use $a(n,t) = \frac{1}{2} \prod_{j=n}^\infty (2a(j,t)) / \prod_{j=n+1}^\infty (2a(j,t))$. Now set
$x=\frac{n}{t}$ and observe $\theta_0(\frac{n+1}{t}) = \theta_0(x+\frac{1}{t})=  \theta_0(x) \pm \theta'(x) \frac{1}{t} + O(t^{-2})$
uniformly in $|x| \le 1 - C$. Similarly for for the other terms and hence on checks that the only difference up to $O(t^{-\alpha})$
errors in the above formulas for $n$ and $n\pm 1$ is a $\mp 2\theta_0$ in the argument of the cosine (stemming from the $t\Phi_0(z_0)$ term).
The second formula follows in the same manner from $b(n,t)= \sum_{j=n}^\infty b(j,t) -\sum_{j=n+1}^\infty b(j,t)$.
\end{proof}

\begin{figure}
\includegraphics[width=8cm]{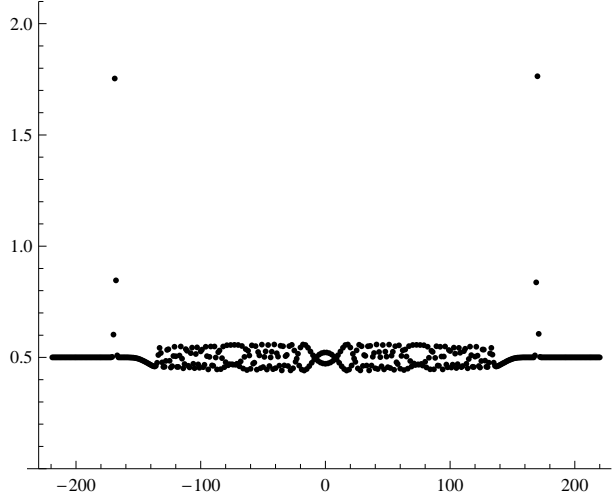}
\caption{Numerically computed solution $a(n,150)$ of the Toda lattice in Flaschka's variables.} \label{fig2}
\end{figure}
This is illustrated in Figure~\ref{fig2}, which shows the same solution as in Figure~\ref{fig1} but in Flaschka's
variables. It is also interesting to look at the relation between the energy $\lam$ of the underlying
Lax operator $H$ and the propagation speed at which the corresponding parts of the Toda
lattice travel, that is, the analog of the classical dispersion relation. By the above theorems,
the nonlinear dispersion relation is given by (see Figure~\ref{fig3})
\begin{figure}
\includegraphics[width=8cm]{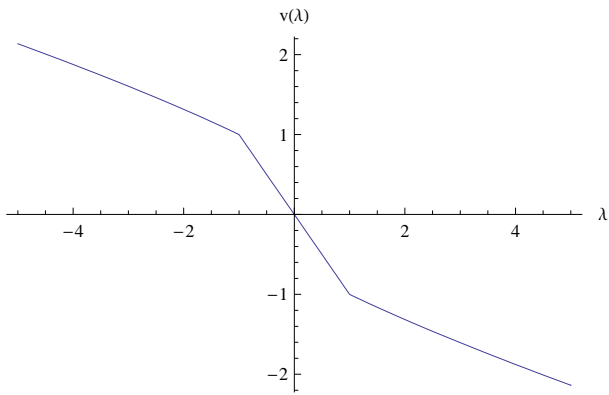}
\caption{Nonlinear dispersion relation for the Toda lattice.} \label{fig3}
\end{figure}
\be
v(\lam)=\frac{n}{t},
\ee
where
\be\label{def:v}
v(\lam) = \begin{cases}
-\lam, & \lam\in[-1,1],\\
\frac{\sqrt{\lam^2-1}}{\log(|\lam-\sqrt{\lam^2-1}|)}, & \lam\in(-\infty,-1]\cup[1,\infty).
\end{cases}
\ee
We will not address the asymptotics in the missing region around $|n| \approx t$.
In the case $|R_+(z,0)|<1$ the solution can be given in terms of Painlev\'e II transcendents.
If $|R_+(z,0)|=1$ (which is the generic case), an additional region, the collisionless shock region,
will appear where the solution can be described in terms of elliptic functions. For the 
Painlev\'e region we refer to \cite{dz},  \cite{km}. For the collisionless shock region
an outline using the $g$-function method was given in \cite{dvz}
(for the case of the Korteweg--de Vires equation). The case of the Toda lattice will be dealt with in
\cite{mnt}.

We also remark that the present methods can also be used to obtain
further terms in the asymptotic expansion \cite{dz2}.

Finally, note that one can obtain the asymptotics for $n \geq 0$ from the
ones for $n\leq 0$ by virtue of a simple reflection. Similarly for $t\geq 0$
versus $t\le 0$.

\begin{lemma}
Suppose $a(n,t)$, $b(n,t)$ satisfy the Toda equation \eqref{tl}, then so do
\[
\ti{a}(n,t) =a(-n-1,t),\quad \ti{b}(n,t) = - b(-n,t)
\]
respectively
\[
\ti{a}(n,t) =a(n,-t),\quad \ti{b}(n,t) = - b(n,-t).
\]
\end{lemma}

\section{The Inverse scattering transform and the Riemann--Hilbert problem}
\label{sec:istrhp}

In this section we want to derive the Riemann--Hilbert problem from scattering theory.
The special case without eigenvalues was first given in Kamvissis \cite{km}. How eigenvalues
can be added was first shown in Deift, Kamvissis, Kriecherbauer, and Zhou \cite{dkkz}.
We essentially follow \cite{krt} in this section.

For the necessary results from scattering theory respectively the inverse
scattering transform for the Toda lattice we refer to \cite{tist}, \cite{tivp}, \cite{tjac}.

Associated with $a(t), b(t)$ is a self-adjoint Jacobi operator
\begin{equation} \label{defjac}
H(t) = a(t) S^+  + a^-(t) S^-  + b(t)
\end{equation}
in $\lz$, where $S^\pm f(n) = f^\pm(n)= f(n\pm1)$ are the usual shift operators and
$\lz$ denotes the Hilbert space of square summable (complex-valued) sequences
over $\Z$. By our assumption \eqref{decay} the spectrum of $H$ consists of an absolutely
continuous part $[-1,1]$ plus a finite number of eigenvalues $\lam_k\in\R\backslash[-1,1]$,
$1\le k \le N$. In addition, there exist two Jost functions $\psi_\pm(z,n,t)$
which solve the recurrence equation
\be
H(t) \psi_\pm(z,n,t) = \frac{z+z^{-1}}{2} \psi_\pm(z,n,t), \qquad |z|\le 1,
\ee
and asymptotically look like the free solutions
\be
\lim_{n \to \pm \infty} z^{\mp  n} \psi_{\pm}(z,n,t) =1.
\ee
Both $\psi_\pm(z,n,t)$ are analytic for $0<|z|<1$ with smooth boundary values
for $|z|=1$.
The asymptotics of the two Jost function are
\be\label{eq:psiasym}
\psi_\pm(z,n,t) =  \frac{z^{\pm n}}{A_\pm(n,t)} \Big(1 + 2 B_\pm(n,t) z + O(z^2) \Big),
\ee
as $z \to 0$, where
\be \label{defAB}
\aligned
A_+(n,t) &= \prod_{j=n}^{\infty} 2 a(j,t), \quad
B_+(n,t)= -\sum_{j=n+1}^\infty b(j,t), \\
A_-(n,t) &= \!\!\prod_{j=- \infty}^{n-1}\! 2 a(j,t), \quad
B_-(n,t) = -\sum_{j=-\infty}^{n-1} b(j,t).
\endaligned
\ee

One has the scattering relations
\be \label{relscat}
T(z) \psi_\mp(z,n,t) =  \ol{\psi_\pm(z,n,t)} +
R_\pm(z,t) \psi_\pm(z,n,t),  \qquad |z|=1,
\ee
where $T(z)$, $R_\pm(z,t)$ are the transmission respectively reflection coefficients.
The transmission and reflection coefficients have the following well-known properties (\cite[Sect.~10.2]{tjac}):

\begin{lemma}
The transmission coefficient $T(z)$ has a meromorphic extension to the
interior of the unit circle with simple poles at the images of the eigenvalues $\zeta_j$.
The residues of $T(z)$ are given by
\be\label{eq:resT}
\res_{\zeta_k} T(z) = - \zeta_k \frac{\gam_{+,k}(t)}{\mu_k(t)} = - \zeta_k \gam_{-,k}(t) \mu_k(t),
\ee
where
\be
\gam_{\pm,k}(t)^{-1} = \sum_{n\in\Z} |\psi_\pm(\zeta_k,n,t)|^2
\ee
and $\psi_- (\zeta_k,n,t) = \mu_k(t) \psi_+(\zeta_k,n,t)$.

Moreover,
\be \label{reltrpm}
T(z) \ol{R_+(z,t)} + \ol{T(z)} R_-(z,t)=0, \qquad |T(z)|^2 + |R_\pm(z,t)|^2=1.
\ee
\end{lemma}

In particular one reflection coefficient, say $R(z,t)=R_+(z,t)$, and one set of
norming constants, say $\gam_k(t)= \gam_{+,k}(t)$, suffices. Moreover,
the time dependence is given by (\cite[Thm.~13.4]{tjac}):

\begin{lemma}
The time evolutions of the quantities $R_+(z,t)$, $\gam_{+,k}(t)$ are given by
\begin{align}
R(z,t) &= R(z) \E^{t (z - z^{-1})}\\
\gam_k(t) &= \gam_k \E^{t (\zeta_k - \zeta_k^{-1})},
\end{align}
where $R(z)=R(z,0)$ and $\gam_k=\gam_k(0)$.
\end{lemma}

Now we define the sectionally meromorphic vector
\be\label{defm}
m(z,n,t)= \left\{\begin{array}{c@{\quad}l}
\begin{pmatrix} T(z) \psi_-(z,n,t) z^n  & \psi_+(z,n,t) z^{-n} \end{pmatrix},
& |z|<1,\\
\begin{pmatrix} \psi_+(z^{-1},n,t) z^n & T(z^{-1}) \psi_-(z^{-1},n,t) z^{-n} \end{pmatrix},
& |z|>1.
\end{array}\right.
\ee
We are interested in the jump condition of $m(z,n,t)$ on the unit circle $\T$ (oriented
counterclockwise). To formulate our jump condition we use the following convention:
When representing functions on $\T$, the lower subscript denotes
the non-tangential limit from different sides,
\be
m_\pm(z) = \lim_{ \zeta\to z,\; |\zeta|^{\pm 1}<1} m(\zeta), \qquad |z|=1.
\ee
In general, for an oriented contour $\Sigma$, $m_+(z)$ (resp.\ $m_-(z)$) will denote the limit
of $m(\zeta)$ as $\zeta\to z$ from the positive (resp.\ negative) side of $\Sigma$. Here
the positive (resp.\ negative) side is the one which lies to the left (resp.\ right) as one traverses the contour in the
direction of the orientation. Using the notation above implicitly assumes that these limits exist in the sense that
$m(z)$ extends to a continuous function on the boundary.

\begin{theorem}[Vector Riemann--Hilbert problem]\label{thm:vecrhp}
Let $\mathcal{S}_+(H(0))=\{ R(z),\; |z|=1; \: (\zeta_k, \gam_k), \: 1\le k \le N \}$
the left scattering data of the operator $H(0)$. Then $m(z)=m(z,n,t)$ defined in \eqref{defm}
is a solution of the following vector Riemann--Hilbert problem.

Find a function $m(z)$ which is meromorphic away from the unit circle with simple poles at
$\zeta_k$, $\zeta_k^{-1}$ and satisfies:
\begin{enumerate}
\item The jump condition
\be \label{eq:jumpcond}
m_+(z)=m_-(z) v(z), \qquad
v(z)=\begin{pmatrix}
1-|R(z)|^2 & - \ol{R(z)} \E^{-t\Phi(z)} \\
R(z) \E^{t\Phi(z)} & 1
\end{pmatrix},
\ee
for $z \in\T$,
\item
the pole conditions
\be\label{eq:polecond}
\aligned
\res_{\zeta_k} m(z) &= \lim_{z\to\zeta_k} m(z)
\begin{pmatrix} 0 & 0\\ - \zeta_k \gam_k \E^{t\Phi(\zeta_k)}  & 0 \end{pmatrix},\\
\res_{\zeta_k^{-1}} m(z) &= \lim_{z\to\zeta_k^{-1}} m(z)
\begin{pmatrix} 0 & \zeta_k^{-1} \gam_k \E^{t\Phi(\zeta_k)} \\ 0 & 0 \end{pmatrix},
\endaligned
\ee
\item
the symmetry condition
\be \label{eq:symcond}
m(z^{-1}) = m(z) \sigI
\ee
\item
and the normalization
\be\label{eq:normcond}
m(0) = (m_1\quad m_2),\quad m_1 \cdot m_2 = 1\quad m_1 > 0.
\ee
\end{enumerate}
Here the phase is given by
\begin{equation}
\Phi(z)=z-z^{-1}+2 \frac{n}{t} \log \, z.
\end{equation}
\end{theorem}

\begin{proof}
The jump condition \eqref{eq:jumpcond} is a simple calculation using the scattering relations
\eqref{relscat} plus \eqref{reltrpm}. The pole conditions follow since $T(z)$ is meromorphic in $|z| <1$
with simple poles at $\zeta_k$ and residues given by \eqref{eq:resT}.
The symmetry condition holds by construction and the normalization \eqref{eq:normcond}
is immediate from the following lemma.
\end{proof}

Observe that the pole condition at $\zeta_k$ is sufficient since the one at $\zeta_k^{-1}$ follows
by symmetry. Moreover, it can be shown that the solution of the above Riemann--Hilbert problem
is unique \cite{krt}. However, we will not need this fact here and it will follow as a byproduct of
our analysis at least for sufficiently large $t$.

Moreover, we have the following asymptotic behaviour near $z=0$:

\begin{lemma}
The function $m(z,n,t)$ defined in \eqref{defm} satisfies
\be\label{eq:AB}
m(z,n,t) = \begin{pmatrix}
A(n,t) (1 - 2 B(n-1,t) z) &
\frac{1}{A(n,t)}(1 + 2 B(n,t) z )
\end{pmatrix} + O(z^2).
\ee
Here $A(n,t)= A_+(n,t)$ and $B(n,t)= B_+(n,t)$ are defined in \eqref{defAB}.
\end{lemma}

\begin{proof}
This follows from \eqref{eq:psiasym} and $T(z)= A_+ A_- ( 1 - 2(B_+ - b +B_-)z+ O(z^2))$.
\end{proof}

For our further analysis it will be convenient to rewrite the pole condition as a jump
condition and hence turn our meromorphic Riemann--Hilbert problem into a holomorphic Riemann--Hilbert problem following \cite{dkkz}.
Choose $\eps$ so small that the discs $|z-\zeta_k|<\eps$ are inside the set $\{z| 0<|z|<1\}$ and
do not intersect. Then redefine $m$ in a neighborhood of $\zeta_k$ respectively $\zeta_k^{-1}$ according to
\be\label{eq:redefm}
m(z) = \begin{cases} m(z) \begin{pmatrix} 1 & 0 \\
\frac{\zeta_k \gamma_k \E^{t\Phi(\zeta_k)} }{z-\zeta_k} & 1 \end{pmatrix},  &
|z-\zeta_k|< \eps,\\
m(z) \begin{pmatrix} 1 & -\frac{z \gamma_k \E^{t\Phi(\zeta_k)} }{z-\zeta_k^{-1}} \\
0 & 1 \end{pmatrix},  &
|z^{-1}-\zeta_k|< \eps,\\
m(z), & \text{else}.\end{cases}
\ee
Then a straightforward calculation using $\res_\zeta m = \lim_{z\to\zeta} (z-\zeta)m(z)$ shows

\begin{lemma}\label{lem:pctoj}
Suppose $m(z)$ is redefined as in \eqref{eq:redefm}. Then $m(z)$ is holomorphic away from
the unit circle and satisfies \eqref{eq:jumpcond}, \eqref{eq:symcond}, \eqref{eq:normcond}
and the pole conditions are replaced by the jump conditions
\be \label{eq:jumpcond2}
\aligned
m_+(z) &= m_-(z) \begin{pmatrix} 1 & 0 \\
\frac{\zeta_k \gamma_k \E^{t\Phi(\zeta_k)}}{z-\zeta_k} & 1 \end{pmatrix},\quad |z-\zeta_k|=\eps,\\
m_+(z) &= m_-(z) \begin{pmatrix} 1 & \frac{z \gamma_k \E^{t\Phi(\zeta_k)}}{z-\zeta_k^{-1}} \\
0 & 1 \end{pmatrix},\quad |z^{-1}-\zeta_k|=\eps,
\endaligned
\ee
where the small circle around $\zeta_k$ is oriented counterclockwise and the one around
$\zeta_k^{-1}$ is oriented clockwise.
\end{lemma}

Finally, we note that the case of just one eigenvalue and zero reflection coefficient can be
solved explicitly.

\begin{lemma}[One soliton solution]\label{lem:singlesoliton}
Suppose there is only one eigenvalue and a vanishing reflection coefficient, that is,
$\mathcal{S}_+(H(t))=\{ R(z)\equiv 0,\; |z|=1; \: (\zeta, \gam) \}$ with $\zeta\in(-1,0)\cup(0,1)$ and $\gam\ge0$.
Then the Riemann--Hilbert problem \eqref{eq:jumpcond}--\eqref{eq:normcond} has a unique solution
is given by
\begin{align}\label{eq:oss}
m_0(z) &= \begin{pmatrix} f(z) & f(1/z) \end{pmatrix} \\
\nn f(z) &= \frac{1}{\sqrt{1 - \zeta^2 + \gam(n,t)} \sqrt{1 - \zeta^2 + \zeta^2 \gam(n,t)}}
\left(\gam(n,t) \zeta^2 \frac{z-\zeta^{-1}}{z - \zeta} + 1 - \zeta^2\right),
\end{align}
where $\gam(n,t)=\gam \E^{t\Phi(\zeta)}$.
In particular,
\be
A_+(n,t) = \sqrt{\frac{1-\zeta^2 + \gam(n,t)}{1 - \zeta^2 + \gam(n,t) \zeta^2}}, \qquad
B_+(n,t) = \frac{\gam(n,t)  \zeta (\zeta ^2-1)}{2 (1 - \zeta^2 + \gam(n,t) \zeta^2)}.
\ee
Furthermore, the zero solution is the only solution of the corresponding vanishing problem where
the normalization is replaced by $m(0) = (0\quad m_2)$ with $m_2$ arbitrary.
\end{lemma}

\begin{proof}
By symmetry, the solution must be of the form $m_0(z) = \begin{pmatrix} f(z) & f(1/z) \end{pmatrix}$,
where $f(z)$ is meromorphic in $\C\cup\{\infty\}$ with the only possible pole at $\zeta$. Hence
\[
f(z) = \frac{1}{A} \left( 1+ 2 \frac{B}{z - \zeta}\right),
\]
where the unknown constants $A$ and $B$ are uniquely determined by the pole condition
$\res_\zeta f(z) = -\zeta \gam(n,t) f(\zeta^{-1})$ and the normalization $f(0) f(\infty)=1$, $f(0)>0$.
\end{proof}

\section{Conjugation and deformation}
\label{sec:conjdef}

This section demonstrates how to conjugate our Riemann--Hilbert problem and
deform the jump contours, such that the jumps will be exponentially close to the identity
away from the stationary phase points. In order to do this we will assume that $R(z)$
has an analytic extension to a strip around the unit circle throughout this and the
following section. This is for example the case if the decay in \eqref{decay} is exponentially.
We will eventually show how to remove this assumption in Section~\ref{sec:analapprox}.

For easy reference we note the following result which can be checked by a straightforward
calculation.

\begin{lemma}[Conjugation]\label{lem:conjug}
Assume that $\wti{\Sigma}\subseteq\Sigma$. Let $D$ be a matrix of the form
\be
D(z) = \begin{pmatrix} d(z)^{-1} & 0 \\ 0 & d(z) \end{pmatrix},
\ee
where $d: \C\backslash\wti{\Sigma}\to\C$ is a sectionally analytic function. Set
\be
\ti{m}(z) = m(z) D(z),
\ee
then the jump matrix transforms according to
\be
\ti{v}(z) = D_-(z)^{-1} v(z) D_+(z).
\ee
If $d$ satisfies $d(z^{-1}) = d(z)^{-1}$ and $d(0) > 0$. Then the transformation $\ti{m}(z) = m(z) D(z)$
respects our symmetry, that is, $\ti{m}(z)$ satisfies \eqref{eq:symcond} if and only if $m(z)$ does.
\end{lemma}

In particular, we obtain
\be
\ti{v} = \begin{pmatrix} v_{11} & v_{12} d^{2} \\ v_{21} d^{-2}  & v_{22} \end{pmatrix},
\qquad z\in\Sigma\backslash\wti{\Sigma},
\ee
respectively
\be
\ti{v} = \begin{pmatrix} \frac{d_-}{d_+} v_{11} & v_{12} d_+ d_- \\
v_{21} d_+^{-1} d_-^{-1}  & \frac{d_+}{d_-} v_{22} \end{pmatrix},
\qquad z\in\wti{\Sigma}.
\ee

In order to remove the poles there are two cases to distinguish. If $\re(\Phi(\zeta_k))<0$
the corresponding jumps \eqref{eq:jumpcond2} are exponentially close to the identity as $t\to\infty$
and there is nothing to do.
Otherwise, if $\re(\Phi(\zeta_k))<0$, we use conjugation to turn the jumps into exponentially decaying
ones, again following Deift, Kamvissis, Kriecherbauer, and Zhou \cite{dkkz} (see also \cite{krt}).
For this purpose we will use the next lemma which shows how $\gam_k \E^{t \Phi(\zeta_k)}$ can be
replaced by its inverse. It turns out that we will have to handle the poles at $\zeta_k$ and $\zeta_k^{-1}$
in one step in order to preserve symmetry and in order to not add additional poles
elsewhere.

\begin{lemma}\label{lem:twopolesinc}
Assume that the Riemann--Hilbert problem for $m$ has jump conditions near $\zeta$ and
$\zeta^{-1}$ given by
\be
\aligned
m_+(z)&=m_-(z)\begin{pmatrix}1&0\\ \frac{\gam \zeta}{z-\zeta}&1\end{pmatrix}, && |z-\zeta|=\eps, \\
m_+(z)&=m_-(z)\begin{pmatrix}1&\frac{\gam z}{z-\zeta^{-1}}\\0&1\end{pmatrix}, && |z^{-1}- \zeta|=\eps.
\endaligned
\ee
Then this Riemann--Hilbert problem is equivalent to a Riemann--Hilbert problem for $\ti{m}$ which has
jump conditions near $\zeta$ and $\zeta^{-1}$ given by
\begin{align*}
\ti{m}_+(z)&= \ti{m}_-(z)\begin{pmatrix}1& \frac{(\zeta z-1)^2}{\zeta (z-\zeta) \gam}\\ 0 &1\end{pmatrix},
&& |z-\zeta|=\eps, \\
\ti{m}_+(z)&= \ti{m}_-(z)\begin{pmatrix}1& 0 \\ \frac{(z-\zeta)^2}{\zeta z (\zeta z-1) \gam}&1\end{pmatrix},
&& |z^{-1}- \zeta|=\eps,
\end{align*}
and all remaining data conjugated (as in Lemma~\ref{lem:conjug}) by
\be
D(z) = \begin{pmatrix} \frac{z - \zeta}{\zeta z-1} & 0 \\ 0 & \frac{\zeta z-1}{z-\zeta} \end{pmatrix}.
\ee
\end{lemma}

\begin{proof}
To turn $\gam$ into $\gam^{-1}$, introduce $D$ by
\[
D(z) = \begin{cases}
\begin{pmatrix} 1 & \frac{1}{\gam} \frac{z-\zeta}{\zeta}\\ - \gam \frac{\zeta}{z-\zeta} & 0 \end{pmatrix}
\begin{pmatrix} \frac{z - \zeta}{\zeta z-1} & 0 \\ 0 & \frac{\zeta z-1}{z-\zeta} \end{pmatrix}, &  |z-\zeta|<\eps, \\
\begin{pmatrix} 0 & \gam \frac{z \zeta}{z \zeta -1} \\ -\frac{1}{\gam} \frac{z \zeta -1}{z \zeta} & 1 \end{pmatrix}
\begin{pmatrix} \frac{z - \zeta}{\zeta z-1} & 0 \\ 0 & \frac{\zeta z-1}{z-\zeta} \end{pmatrix}, & |z^{-1}-\zeta|<\eps, \\
\begin{pmatrix} \frac{z - \zeta}{\zeta z-1} & 0 \\ 0 & \frac{\zeta z-1}{z-\zeta} \end{pmatrix}, & \text{else},
\end{cases}
\]
and note that $D(z)$ is analytic away from the two circles. Now set $\ti{m}(z) = m(z) D(z)$, which is again
symmetric by $D(z^{-1})= \ssigI D(z) \ssigI$.
The jumps along $|z-\zeta|=\eps$ and $|z^{-1}- \zeta|=\eps$ follow by a straightforward calculation and
the remaining jumps follow from Lemma~\ref{lem:conjug}.
\end{proof}

The jumps along $\T$ are of oscillatory type and our aim is to apply a contour deformation which will
move them into regions where the oscillatory terms will decay exponentially. Since the
jump matrix $v$ contains both $\exp(t \Phi)$ and $\exp(-t \Phi)$ we need to separate them
in order to be able to move them to different regions of the complex plane. For this
we will need the following factorizations of the jump condition \eqref{eq:jumpcond}.
First of all
\be
v(z)= b_-(z)^{-1} b_+(z),
\ee
where
\[
b_-(z)= \begin{pmatrix}
1 & \ol{R(z)}  \E^{-t\Phi(z)} \\ 0 &1
\end{pmatrix}, \qquad
b_+(z)= \begin{pmatrix}
1 & 0 \\ R(z) \E^{t\Phi(z)} & 1
\end{pmatrix}.
\]
This will be the proper factorization for $z>z_0$. Here $z>z_0$ has to be understood as $\lam(z)>\lam_0$.
Similarly, we have
\be\label{facB}
v(z)= B_-(z)^{-1} \begin{pmatrix} 1-|R(z)|^2 & 0 \\ 0 & \frac{1}{1-|R(z)|^2}\end{pmatrix} B_+(z),
\ee
where
\[
B_-(z) =\begin{pmatrix}
1 & 0 \\ - \frac{R(z)  \E^{t\Phi(z)}}{1-|R(z)|^2} &1
\end{pmatrix}, \qquad
B_+(z)= \begin{pmatrix}
1 & -\frac{\ol{R(z)}  \E^{-t\Phi(z)}}{1-|R(z)|^2} \\ 0 & 1
\end{pmatrix}.
\]
This will be the proper factorization for $z<z_0$.

To get rid of the diagonal part we need to solve the corresponding scalar Riemann--Hilbert problem.
Moreover, for $z_0\in(-1,0)$ we have $\re(\Phi(z))>0$ for $z\in(\zeta_0,0)$ and $\re(\Phi(z))<0$
for $z\in(-1,\zeta_0)\cup(0,1)$, for $z_0\in\T$ we have $\re(\Phi(z))>0$ for $z\in(-1,0)$ and
$\re(\Phi(z))<0$ for $z\in(0,1)$, and for $z_0\in(0,1)$ we have $\re(\Phi(z))>0$ for
$z\in(-1,0)\cup(\zeta_0,1)$ and $\re(\Phi(z))<0$ for $z\in(0,\zeta_0)$ (compare Figure~\ref{fig:signRePhi}
and note that by $\re(\Phi(z^{-1}))= -\re(\Phi(z))$ the curves $\re(\Phi(z))=0$ are symmetric
with respect to $z\mapsto z^{-1}$).

Together with the Blaschke factors needed to conjugate the jumps near the eigenvalues,
this is just the partial transmission coefficient $T(z,z_0)$ introduced in \eqref{def:Tzz0}. In fact, it
satisfies the following scalar meromorphic Riemann--Hilbert problem:

\begin{lemma}\label{lemT}
Set $\Sigma(z_0)=\emptyset$ for $z_0\in(-1,0)$,  $\Sigma(z_0)= \{z \in\T | \re(z)<\re(z_0)\}$ for
$z_0\in\T$, and $\Sigma(z_0)=\T$ for $z_0\in(0,1)$. Then the partial transmission coefficient
$T(z,z_0)$ is meromorphic for $z\in\C\backslash\Sigma(z_0)$, with
simple poles at $\zeta_j$ and simple zeros at $\zeta_j^{-1}$ for all $j$ with
$\frac{1}{2}(\zeta_j+\zeta_j^{-1})<\lam_0$, and satisfies the jump condition
\[
T_+(z,z_0) = T_-(z,z_0) (1 - |R(z)|^2), \qquad z\in\Sigma(z_0).
\]
Moreover,
\begin{enumerate}
\item
$T(z^{-1},z_0) = T(z,z_0)^{-1}$, $z\in\C\backslash\Sigma(z_0)$, and $T(0,z_0)>0$,
\item
$\ol{T(z,z_0)}=T(\ol{z},z_0)$, $z\in\C$ and, in particular, $T(z,z_0)$ is real-valued for $z\in\R$,
\item
$T(z,z_0) = T(z) (C+o(1))$ with $C\ne 0$ for $|z|\le 1$ near $\pm 1$ if $\pm 1 \in\Sigma(z_0)$
and continuous otherwise.
\end{enumerate}
\end{lemma}

\begin{proof}
That $\zeta_j$ are simple poles and $\zeta_j^{-1}$ are simple zeros is obvious from the Blaschke
factors and that $T(z,z_0)$ has the given jump follows from Plemelj's formulas.
(i)--(iii) are straightforward to check.
\end{proof}
Observe that for $\zeta_0 < \zeta_N$ if $\zeta_N\in(0,1)$ respectively $\zeta_0 < 1$ else we have
$T(z)=T(z,z_0)$.

Moreover, note that (i) and (ii) imply
\be\label{absparT}
|T(z,z_0)|^2 = T(\ol{z},z_0) T(z,z_0)  = T(z^{-1},z_0) T(z,z_0) = 1, \qquad z\in\T\backslash\Sigma(z_0).
\ee
Now we are ready to perform our conjugation step. Introduce
\[
D(z) = \begin{cases}
\begin{pmatrix} 1 & \frac{z-\zeta_k}{\zeta_k \gam_k \E^{t\Phi(\zeta_k)}}\\
- \frac{\zeta_k \gam_k \E^{t\Phi(\zeta_k)}}{z-\zeta_k} & 0 \end{pmatrix}
D_0(z), &  |z-\zeta_k|<\eps, \: \lam_k < \frac{1}{2}(\zeta_0+\zeta_0^{-1}),\\
\begin{pmatrix} 0 & \frac{z \zeta_k \gam_k \E^{t\Phi(\zeta_k)}}{z \zeta_k -1} \\
-\frac{z \zeta_k -1}{z \zeta_k \gam_k \E^{t\Phi(\zeta_k)}} & 1 \end{pmatrix}
D_0(z), & |z^{-1}-\zeta_k|<\eps, \: \lam_k < \frac{1}{2}(\zeta_0+\zeta_0^{-1}),\\
D_0(z), & \text{else},
\end{cases}
\]
where
\[
D_0(z) = \begin{pmatrix} T(z,z_0)^{-1} & 0 \\ 0 & T(z,z_0) \end{pmatrix}.
\]
Note that we have
\[
D(z^{-1})= \sigI D(z) \sigI.
\]
Now we conjugate our vector $m(z)$ defined in \eqref{defm} respectively \eqref{eq:redefm}
using $D(z)$,
\be\label{eq:tim}
\ti{m}(z)=m(z) D(z).
\ee
Since, by Lemma~\ref{lemT} (iii), $T(z,z_0)$ is either nonzero and continuous near $z=\pm1$ (if $\pm 1 \notin\Sigma(z_0)$) or it has
the same behaviour as $T(z)$ near $z=\pm 1$ (if $\pm 1 \in\Sigma(z_0)$), the new vector $\ti{m}(z)$
is again continuous near $z=\pm 1$ (even if $T(z)$ vanishes there).

Then using Lemma~\ref{lem:conjug} and Lemma~\ref{lem:twopolesinc} the jumps
corresponding to eigenvalues $\lam_k <\frac{1}{2}(\zeta_0+\zeta_0^{-1})$ (if any) are given by
\be
\aligned
\ti{v}(z) &= \begin{pmatrix}1& \frac{z-\zeta_k}{\zeta_k
\gam_k T(z,z_0)^{-2} \E^{t\Phi(\zeta_k)} }\\ 0 &1\end{pmatrix},
\qquad |z-\zeta_k|=\eps, \\ \label{eq:tvzetak1}
\ti{v}(z) &= \begin{pmatrix}1& 0 \\ \frac{\zeta_k z -1}{\zeta_k z \gam_k T(z,z_0)^2
\E^{t\Phi(\zeta_k)}}&1\end{pmatrix},
\qquad |z^{-1}- \zeta_k|=\eps,
\endaligned
\ee
and corresponding to eigenvalues $\lam_k > \frac{1}{2}(\zeta_0+\zeta_0^{-1})$ (if any) by
\be
\aligned
\ti{v}(z) &= \begin{pmatrix} 1 & 0 \\ \frac{\zeta_k \gam_k T(z,z_0)^{-2} \E^{t\Phi(\zeta_k)}}{z-\zeta_k}
 & 1 \end{pmatrix},
\qquad |z-\zeta_k|=\eps, \\ \label{eq:tvzetak2}
\ti{v}(z) &= \begin{pmatrix} 1 & \frac{z \gam_k T(z,z_0)^2 \E^{t\Phi(\zeta_k)}}{z-\zeta_k^{-1}} \\
0 & 1 \end{pmatrix},
\qquad |z^{-1}-\zeta_k|=\eps.
\endaligned
\ee
In particular, an investigation of the sign of $\re(\Phi(z))$ (see Figure~\ref{fig:signRePhi} below) shows
that all off-diagonal entries of these jump matrices, except for possibly one if
$\zeta_{k_0}=\zeta_0$ for some $k_0$, are exponentially decreasing. In the latter case we will keep the
pole condition for $\zeta_{k_0}=\zeta_0$ which now reads
\be
\aligned
\res_{\zeta_{k_0}} \ti{m}(z) &= \lim_{z\to\zeta_{k_0}} \ti{m}(z)
\begin{pmatrix} 0 & 0\\ - \zeta_{k_0} \gam_{k_0} T(\zeta_{k_0},z_0)^{-2} \E^{t\Phi(\zeta_{k_0})}  & 0 \end{pmatrix},\\\label{eq:tvzetak0}
\res_{\zeta_{k_0}^{-1}} \ti{m}(z) &= \lim_{z\to\zeta_{k_0}^{-1}} \ti{m}(z)
\begin{pmatrix} 0 & \zeta_{k_0}^{-1} \gam_{k_0} T(\zeta_{k_0},z_0)^{-2} \E^{t\Phi(\zeta_{k_0})} \\ 0 & 0 \end{pmatrix}.
\endaligned
\ee
Furthermore, the jump along
$\T$ is given by
\be
\ti{v}(z) = \begin{cases}
\ti{b}_-(z)^{-1} \ti{b}_+(z), \qquad \lam(z)> \lam_0,\\
\ti{B}_-(z)^{-1} \ti{B}_+(z), \qquad \lam(z)< \lam_0,\\
\end{cases}
\ee
where
\be\label{eq:deftib}
\ti{b}_-(z) = \begin{pmatrix} 1 & \frac{R(z^{-1}) \E^{-t\Phi(z)}}{T(z^{-1},z_0)^2} \\ 0 & 1 \end{pmatrix}, \quad
\ti{b}_+(z) = \begin{pmatrix} 1 & 0 \\ \frac{R(z) \E^{t\Phi(z)}}{T(z,z_0)^2}& 1 \end{pmatrix},
\ee
and
\begin{align}\nn
\ti{B}_-(z) &= \begin{pmatrix} 1 & 0 \\ - \frac{T_-(z,z_0)^{-2}}{1-R(z)R(z^{-1})} R(z) \E^{t\Phi(z)} & 1 \end{pmatrix}, \\
\ti{B}_+(z) &= \begin{pmatrix} 1 & - \frac{T_+(z,z_0)^2}{1-R(z)R(z^{-1})} R(z^{-1}) \E^{-t\Phi(z)} \\ 0 & 1 \end{pmatrix}.
\end{align}
Here we have used $T_\pm(z^{-1},z_0)=T_\pm(\ol{z},z_0)=\ol{T_\pm(z,z_0)}$ and
$R(z^{-1})=R(\ol{z})=\ol{R(z)}$ for $z\in\T$ to show that there exists an analytic continuation
into a neighborhood of the unit circle. Moreover, using
\[
T_\pm(z,z_0)=T_\mp(z^{-1},z_0)^{-1}, \qquad z\in\Sigma(z_0),
\]
we can write
\be \label{eq:relB}
\frac{T_-(z,z_0)^{-2}}{1-R(z)R(z^{-1})} = \frac{\ol{T_-(z,z_0)}}{T_-(z,z_0)}, \quad
\frac{T_+(z,z_0)^2}{1-R(z)R(z^{-1})} = \frac{T_+(z,z_0)}{\ol{T_+(z,z_0)}}
\ee
for $z\in\T$, which shows that the matrix entries are in fact bounded.

Now we deform the jump along $\T$ to move the oscillatory terms
into regions where they are decaying. There are three cases to distinguish (see
Figure~\ref{fig:signRePhi}):
\begin{figure}
\begin{picture}(3.5,3)

\put(1.1,2.6){$z_0\in(-1,0)$}
\put(0.2,0.2){$+$}
\put(0,1.3){$-$}
\put(1.2,0.8){$-$}
\put(1.2,1.15){$+$}

\put(0,1.2){\vector(1,0){3.2}}
\put(1.8,0){\vector(0,1){2.5}}
\put(1.017,1.2){\circle*{0.06}}
\put(0.87,1.35){\scriptsize $\zeta_0$}
\put(0.51,1.2){\circle*{0.06}}
\put(0.0,0.95){\scriptsize $\zeta_0^{-1}$}

\closecurve(0.8,1.2, 1.093,0.493, 1.8,0.2, 2.507,0.493, 2.8,1.2, 2.507,1.907,
1.8,2.2, 1.093,1.907)
\closecurve(1.017,1.2, 1.111,1.094, 1.213,1.05, 1.733,1.094, 1.8,1.2,
1.733,1.306, 1.213,1.35, 1.111,1.306)
\curve(-0.05,0.5, 0.064,0.64, 0.19,0.78, 0.322,0.92, 0.452,1.06, 0.523,1.2,
0.452,1.34, 0.322,1.48, 0.19,1.62, 0.064,1.76, -0.05,1.9)
\end{picture}\quad
\begin{picture}(3.5,3)

\put(0.5,2.6){$z_0\in\T$}
\put(0.2,0.2){$-$}
\put(0.7,0.7){$+$}
\put(1.5,1.3){$-$}
\put(2.4,1.5){$+$}

\put(0,1.2){\vector(1,0){3.2}}
\put(1.2,0){\vector(0,1){2.5}}
\put(1.91,1.9){\circle*{0.06}}
\put(1.7,2.1){\scriptsize $z_0^{-1}$}
\put(1.91,0.5){\circle*{0.06}}
\put(1.8,0.25){\scriptsize $z_0$}

\closecurve(0.2,1.2, 0.493,0.493, 1.2,0.2, 1.907,0.493, 2.2,1.2, 1.907,1.907,
1.2,2.2, 0.493,1.907)

\curve(2.359,-0.1, 2.293,0, 2.223,0.1, 2.148,0.2,
2.068,0.3, 1.981,0.4, 1.9,0.5, 1.782,0.6, 1.668,0.7, 1.547,0.8,
1.424,0.9, 1.316,1., 1.236,1.1, 1.2,1.2, 1.236,1.3, 1.316,1.4,
1.424,1.5, 1.547,1.6, 1.668,1.7, 1.782,1.8, 1.914,1.9, 1.981,2.,
2.068,2.1, 2.148,2.2, 2.223,2.3, 2.293,2.4, 2.359,2.5)

\end{picture}\quad
\begin{picture}(3.5,3)

\put(0.5,2.6){$z_0\in(0,1)$}
\put(0.2,0.2){$-$}
\put(0.7,0.7){$+$}
\put(1.5,1.2){$-$}
\put(2.3,1.7){$-$}
\put(2.8,0.9){$+$}

\put(0,1.2){\vector(1,0){3.2}}
\put(1.2,0){\vector(0,1){2.5}}
\put(1.983,1.2){\circle*{0.06}}
\put(1.92,1.33){\scriptsize $\zeta_0$}
\put(2.5,1.2){\circle*{0.06}}
\put(2.69,1.33){\scriptsize $\zeta_0^{-1}$}

\closecurve(0.2,1.2, 0.493,0.493, 1.2,0.2, 1.907,0.493, 2.2,1.2,
1.907,1.907, 1.2,2.2, 0.493,1.907)
\closecurve(1.2,1.2, 1.267,1.094, 1.787,1.05, 1.889,1.094, 1.983,1.2,
1.889,1.306, 1.787,1.35, 1.267,1.306, 1.2,1.2)
\curve(3.053,0.5, 2.905,0.675, 2.745,0.85, 2.578,1.025, 2.477,1.2,
2.578,1.375, 2.745,1.55, 2.905,1.725, 3.053,1.9)
\end{picture}

\caption{Sign of $\re(\Phi(z))$ for different values of $z_0$}\label{fig:signRePhi}
\end{figure}
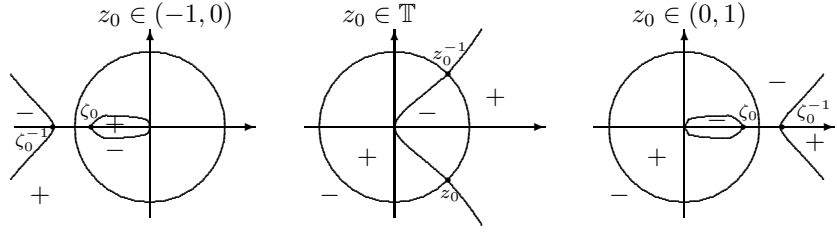

{\bf Case 1: $z_0\in(-1,0)$.}
In this case we will set $\Sigma_\pm=\{ z |\, |z|=(1-\eps)^{\pm 1}\}$ for some small
$\eps\in(0,1)$ such that $\Sigma_\pm$ lies in the region with $\pm \re(\Phi(z))< 0$
and such that we do not intersect the original contours (i.e., we stay away from
$\zeta_j^{\pm 1}$).
Then we can split our jump by redefining $\ti{m}(z)$ according to
\be \label{eq:tib}
\hat{m}(z) = \begin{cases}
\ti{m}(z) \ti{b}_+(z)^{-1}, & (1-\eps)<|z|<1,\\
\ti{m}(z) \ti{b}_-(z)^{-1}, & 1<|z|<(1-\eps)^{-1},\\
\ti{m}(z), & \text{else}.
\end{cases}
\ee
It is straightforward to check that the jump along $\T$ disappears and the
jump along $\Sigma_\pm$ is given by
\be\label{eq:jumpsolreg}
\hat{v}(z) = \begin{cases}
\ti{b}_+(z), & z\in\Sigma_+, \\
\ti{b}_-(z)^{-1}, & z\in\Sigma_-.
\end{cases}
\ee
The other jumps \eqref{eq:tvzetak1}, \eqref{eq:tvzetak2} as well as the pole condition \eqref{eq:tvzetak0}
(if present) are unchanged. 

Note that the resulting Riemann--Hilbert problem still
satisfies our symmetry condition \eqref{eq:symcond} since we have
\[
\ti{b}_\pm(z^{-1}) = \sigI \ti{b}_\mp(z) \sigI .
\]
By construction all jumps \eqref{eq:tvzetak1}, \eqref{eq:tvzetak2}, and \eqref{eq:tib} are exponentially close to
the identity as $t\to\infty$. The only non-decaying part being the pole condition \eqref{eq:tvzetak0}
(if present).

{\bf Case 2: $z_0\in\T\setminus\{\pm1\}$.}
In this case we will set $\Sigma_\pm=\Sigma_\pm^1\cup\Sigma_\pm^2$ as indicated in
Figure~\ref{fig:defcont}. Again note that $\Sigma_\pm^1$ respectively $\Sigma_\mp^2$
lies in the region with $\pm \re(\Phi(z))< 0$ and must be chosen such that we do not intersect any
other parts of the contour.
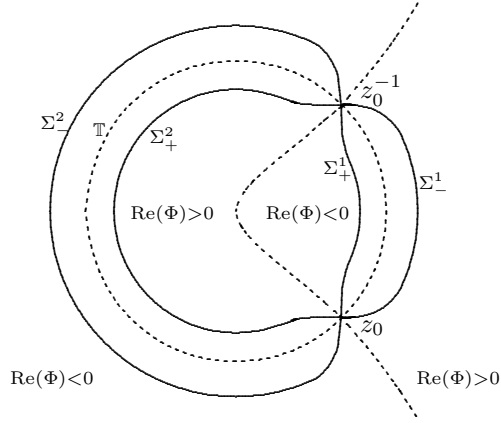
\begin{figure}
\begin{picture}(6,6)

\put(0,0.7){$\scriptstyle \re(\Phi)<0$}
\put(5.4,0.7){$\scriptstyle \re(\Phi)>0$}
\put(1.6,2.9){$\scriptstyle \re(\Phi)>0$}
\put(3.4,2.9){$\scriptstyle \re(\Phi)<0$}

\put(0.4,4.1){$\scriptstyle \Sigma_-^2$}
\put(4.16,3.5){$\scriptstyle \Sigma_+^1$}
\put(1.1,4){$\scriptstyle \T$}
\put(1.85,3.9){$\scriptstyle \Sigma_+^2$}
\put(5.43,3.3){$\scriptstyle \Sigma_-^1$}

\put(4.65,4.5){$z_0^{-1}$}
\put(4.65,1.4){$z_0$}

\curve(1.38,3., 1.4,2.747, 1.459,2.499, 1.557,2.265, 1.689,2.048,
1.854,1.854, 2.048,1.689, 2.265,1.557, 2.499,1.459, 2.747,1.4,
3.,1.38, 3.253,1.4, 3.501,1.459, 3.739,1.55, 4.029,1.584,
4.428,1.572, 4.852,1.654, 5.15,1.904, 5.302,2.252, 5.39,2.621,
5.42,3., 5.39,3.379, 5.302,3.748, 5.15,4.096, 4.852,4.346,
4.428,4.428, 4.029,4.416, 3.739,4.45, 3.501,4.541, 3.253,4.6,
3.,4.62, 2.747,4.6, 2.499,4.541, 2.265,4.443, 2.048,4.311,
1.854,4.146, 1.689,3.952, 1.557,3.735, 1.459,3.501, 1.4,3.253,
1.38,3.)

\curve(0.531,3., 0.561,2.614, 0.652,2.237, 0.8,1.879, 1.002,1.549,
1.254,1.254, 1.549,1.002, 1.879,0.8, 2.237,0.652, 2.614,0.561,
3.,0.531, 3.386,0.561, 3.763,0.652, 4.116,0.809, 4.343,1.151,
4.4,1.6, 4.413,1.973, 4.477,2.247, 4.572,2.489, 4.633,2.741,
4.653,3., 4.633,3.259, 4.572,3.511, 4.477,3.753, 4.413,4.027,
4.4,4.4, 4.343,4.849, 4.116,5.191, 3.763,5.348, 3.386,5.439,
3.,5.469, 2.614,5.439, 2.237,5.348, 1.879,5.2, 1.549,4.998,
1.254,4.746, 1.002,4.451, 0.8,4.121, 0.652,3.763, 0.561,3.386,
0.531,3.)

\curvedashes{0.05,0.05}
\curve(1.,3., 1.586,1.586, 3.,1., 4.414,1.586, 5.,3., 4.414,4.414,
3.,5., 1.586,4.414, 1.,3.)

\curve(5.441,0.2, 5.317,0.4, 5.186,0.6, 5.046,0.8, 4.897,1.,
4.736,1.2, 4.562,1.4, 4.4,1.6, 4.164,1.8, 3.937,2., 3.693,2.2,
3.449,2.4, 3.231,2.6, 3.072,2.8, 3.,3., 3.072,3.2, 3.231,3.4,
3.449,3.6, 3.693,3.8, 3.937,4., 4.164,4.2, 4.428,4.4, 4.562,4.6,
4.736,4.8, 4.897,5., 5.046,5.2, 5.186,5.4, 5.317,5.6, 5.441,5.8)
\end{picture}
\caption{Deformed contour}\label{fig:defcont}
\end{figure}
Then we can split our jump by redefining $\ti{m}(z)$ according to
\be
\hat{m}(z) = \begin{cases}
\ti{m}(z) \ti{b}_+(z)^{-1}, & z \text { between $\T$ and } \Sigma_+^1,\\
\ti{m}(z) \ti{b}_-(z)^{-1}, & z \text { between $\T$ and } \Sigma_-^1,\\
\ti{m}(z) \ti{B}_+(z)^{-1}, & z \text { between $\T$ and } \Sigma_+^2,\\
\ti{m}(z) \ti{B}_-(z)^{-1}, & z \text { between $\T$ and } \Sigma_-^2,\\
\ti{m}(z), & \text{else}.
\end{cases}
\ee
One checks that the jump along $\T$ disappears and the
jump along $\Sigma_\pm$ is given by
\be
\hat{v}(z) = \begin{cases}
\ti{b}_+(z), & z\in\Sigma_+^1, \\
\ti{b}_-(z)^{-1}, & z\in\Sigma_-^1,\\
\ti{B}_+(z), & z\in\Sigma_+^2,\\
\ti{B}_-(z)^{-1}, & z\in\Sigma_-^2.
\end{cases}
\ee
All other jumps \eqref{eq:tvzetak1} and \eqref{eq:tvzetak2} are unchanged. Again the resulting Riemann--Hilbert
problem still satisfies our symmetry condition \eqref{eq:symcond} and the jump along
$\Sigma_\pm$ away from the stationary phase points $z_0$, $z_0^{-1}$ is exponentially
close to the identity as $t\to\infty$.

{\bf Case 3: $z_0\in(0,1)$.}
In this case we will set $\Sigma_\pm=\{ z |\, |z|=(1-\eps)^{\pm 1}\}$ for some small
$\eps\in(0,1)$ such that $\Sigma_\pm$ lies in the region with $\mp \re(\Phi(z))< 0$
and such that we do not intersect the original contours.
Then we can split our jump by redefining $\ti{m}(z)$ according to
\be\label{eq:tiB}
\hat{m}(z) = \begin{cases}
\ti{m}(z) \ti{B}_+(z)^{-1}, & (1-\eps)<|z|<1,\\
\ti{m}(z) \ti{B}_-(z)^{-1}, & 1<|z|<(1-\eps)^{-1},\\
\ti{m}(z), & \text{else}.
\end{cases}
\ee
One checks that the jump along $\T$ disappears and the
jump along $\Sigma_\pm$ is given by
\be
\hat{v}(z) = \begin{cases} \ti{B}_+(z), & z\in\Sigma_+, \\
\ti{B}_-(z)^{-1}, & z\in\Sigma_-.\end{cases}
\ee
The other jumps \eqref{eq:tvzetak1}, \eqref{eq:tvzetak2} as well as the pole condition \eqref{eq:tvzetak0}
(if present) are unchanged. Again the resulting Riemann--Hilbert problem still
satisfies our symmetry condition \eqref{eq:symcond} and all jumps \eqref{eq:tvzetak1}, \eqref{eq:tvzetak2},
and \eqref{eq:tiB} are exponentially close to
the identity as $t\to\infty$. The only non-decaying part being the pole condition \eqref{eq:tvzetak0}
(if present).

In Case~1 and 3 we can immediately apply Theorem~\ref{thm:remcontour} to $\hat{m}$ as follows:
If $|\frac{n}{t} - c_k|>\eps$ for all $k$ we can choose $\gam_0=0$.
Since the error between $\hat{w}^t$ and $\hat{w}_0^t$ is exponentially small, this proves the second
part of Theorem~\ref{thm:asym} in the analytic case upon comparing
\be
m(z) = \hat{m}(z) \begin{pmatrix} T(z,z_0) & 0\\ 0 & T(z,z_0)^{-1} \end{pmatrix}
\ee
with \eqref{eq:AB}. The changes necessary for the
general case will be given in Section~\ref{sec:analapprox}.

Otherwise, if $|\frac{n}{t} - c_k|<\eps$ for some $k$, we choose $\gam_0^t=\gam_k(n,t)$.
Again we conclude that the error between $\hat{w}^t$ and $\hat{w}_0^t$ is
exponentially small, proving the first part of Theorem~\ref{thm:asym}. The changes necessary for the
general case will also be given in Section~\ref{sec:analapprox}.

In Case~2 the jump will not decay on the two small crosses containing the stationary phase
points $z_0$ and $z_0^{-1}$. Hence we will need to continue the investigation of this problem
in the next section.

\section{Reduction to a Riemann--Hilbert problem on a small cross}
\label{sec:reducecross}

In the previous section we have shown that for $z_0\in\T\backslash\{\pm 1\}$ we can
reduce everything to a Riemann--Hilbert problem for $\hat{m}(z)$ such that the jumps
are of order $O(t^{-1})$ except in a small neighborhoods of the stationary phase
points $z_0$ and $z_0^{-1}$. Denote by $\Sigma^C(z_0^{\pm 1})$ the parts of
$\Sigma_+\cup\Sigma_-$ inside a small neighborhood of $z_0^{\pm 1}$. In this section
we will show that everything can reduced to solving the two problems in the two small
crosses $\Sigma^C(z_0)$ respectively $\Sigma^C(z_0^{-1})$.

It will be slightly more convenient to use the alternate normalization
\be\label{eq:checkm}
\check{m}(z) = \frac{1}{\ti{A}} \hat{m}(z), \qquad A = T_0 \ti{A},
\ee
such that
\be
\check{m}(0) = \begin{pmatrix} 1 & \frac{1}{\ti{A}^2} \end{pmatrix}.
\ee
Without loss of generality we can also assume that $\hat{\Sigma}$ consists of two
straight lines in a sufficiently small neighborhood of $z_0$.

We will need the solution of the corresponding $2\times2$ matrix
\be\label{eq:RHPM}
\aligned
M^C_+(z)&= M^C_-(z) \ti{v}(z), \qquad z\in\Sigma^C,\\
M^C(\infty)&=\id,
\endaligned
\ee
where the jump $\ti{v}$ is the same as for $\ti{m}(z)$ but restricted to a
neighborhood of one of the two crosses $\Sigma^C=(\Sigma_+\cup\Sigma_-)\cap
\{ z | \,|z-z_0|<\eps/2\}$ for some small $\eps>0$.

As a first step we make a change of coordinates
\be\label{eqchoc}
\zeta= \frac{\sqrt{-2\sin(\theta_0)}}{z_0 \I} (z-z_0), \qquad
z= z_0 + \frac{z_0 \I}{\sqrt{-2\sin(\theta_0)}} \zeta
\ee
such that the phase reads $\Phi(z)= \I \Phi_0 + \frac{\I}{2} \zeta^2 +O(\zeta^3)$.
Here we have set
\[
z_0 = \E^{\I \theta_0}, \qquad \theta_0\in(-\pi,0),
\]
respectively $\cos(\theta_0) = -n/t$, which implies
\[
\Phi_0 = 2(\sin(\theta_0) - \theta_0 \cos(\theta_0)),\qquad
\Phi''(z_0) = 2 \I\E^{-2\I\theta_0} \sin(\theta_0).
\]
The corresponding Riemann--Hilbert problem will be solved in Section~\ref{sec:cross}.
To apply this result we need the behaviour of our jump matrices near $z_0$,
that is, the behaviour of $T(z,z_0)$ near $z\to z_0$.

\begin{lemma}
Let $z_0\in\T$, then
\be
T(z,z_0) = \left(-\ol{z_0}\frac{z-z_0}{z-\ol{z_0}}\right)^{\I\nu} \ti{T}(z,z_0)
\ee
where $\nu = -\frac{1}{\pi} \log(|T(z_0)|)$ and the branch cut of the logarithm used to define $z^{\I\nu}=\E^{\I\nu\log(z)}$
is chosen along the negative real axis. Here
\[
\ti{T}(z,z_0) = \prod\limits_{\zeta_k\in(-1,0)} |\zeta_k| \frac{z-\zeta_k^{-1}}{z-\zeta_k} \cdot
\exp\left(\frac{1}{2\pi\I}\int\limits_{\ol{z_0}}^{z_0}\log\Big(\frac{|T(s)|}{|T(z_0)|}\Big) \frac{s+z}{s-z} \frac{ds}{s}\right),
\]
is H\"older continuous of any exponent less than $1$ at $z=z_0$ and satisfies $\ti{T}(z_0,z_0)\in\T$.
\end{lemma}

\begin{proof}
This follows since
\[
\exp\left(\frac{1}{2\pi\I}\int\limits_{\ol{z_0}}^{z_0}\log\big(|T(z_0)|\big) \frac{s+z}{s-z} \frac{ds}{s}\right) = \left(-\ol{z_0}\frac{z-z_0}{z-\ol{z_0}}\right)^{\I\nu}.
\]
The property $\ti{T}(z_0,z_0)\in\T$ follows after letting $z\to z_0$ in \eqref{absparT}.
\end{proof}

Now if $z(\zeta)$ is defined as in \eqref{eqchoc} and $0 < \alpha < 1$,
then there is an $L > 0$ such that
\[
|T(z(\zeta), z_0) - \zeta^{\I \nu} \ti{T}(z_0,z_0) \E^{-\frac{3}{2} \I \nu \log(-2\sin(\theta_0))}|
\leq L |\zeta|^{\alpha},
\]
where the branch cut of $\zeta^{\I \nu}$ is tangent to the negative real axis. Clearly
we also have
\[
|R(z(\zeta)) - R(z_0)| \le L  |\zeta|^{\alpha}
\]
and thus the assumptions of Theorem~\ref{thm:solcross} are satisfied with
\[
r = R(z_0) \ti{T}(z_0,z_0)^{-2} \E^{3 \I \nu \log(-2\sin(\theta_0))}
\]
and the solution of \eqref{eq:RHPM} is given by
\begin{align}
M^C(z) &= \id - \frac{z_0}{(-2\sin(\theta_0) t)^{1/2}}  \frac{M_0}{z - z_0} +
O\left(\frac{1}{t^\alpha}\right), \\ \nn
M_0 &= \begin{pmatrix} 0 & -\beta \\  \ol{\beta} & 0 \end{pmatrix},\\
\beta &=\sqrt{\nu} \E^{\I(\pi/4-\arg(R(z_0)) + \arg(\Gamma(\I\nu)))}
(-2\sin(\theta_0))^{-3\I \nu} \ti{T}(z_0, z_0)^2 \E^{-\I t \Phi_0} t^{-\I\nu},
\end{align}
where  $1/2 < \alpha < 1$, and $\cos(\theta_0)=-\lam_0$. Note $|r|=|R(z_0)|$ and
hence $\nu=-\frac{1}{2\pi}\log(1-|R(z_0)|^2)$.

Now we are ready to show

\begin{theorem}\label{thm:decoupling}
The solution $\check{m}(z)$ is given by
\be
\check{m}(z) = \rI - \frac{1}{(-2\sin(\theta_0) t)^{1/2}} ( m_0(z) + \bar{m}_0(z))
+ O\left(\frac{1}{t^\alpha}\right),
\ee
where
\be
m_0(z) = \begin{pmatrix} \ol{\beta} \frac{z}{z-z_0} & -\beta \frac{z_0}{z-z_0} \end{pmatrix}, \qquad
\bar{m}_0(z) = \ol{m_0(\ol{z})} = m_0(z^{-1})\sigI.
\ee
\end{theorem}

\begin{proof}
Introduce $m(z)$ by
\[
m(z) = \begin{cases}
\check{m}(z) M^C(z)^{-1}, & |z-z_0| \le \eps,\\
\check{m}(z) \ti{M}^C(z)^{-1}, & |z^{-1}-z_0| \le \eps,\\
\check{m}(z), & \text{else},\end{cases}
\]
where
\[
\ti{M}^C(z) = \sigI M^C(z^{-1}) \sigI = \id - \frac{z}{(-2\sin(\theta_0) t)^{1/2}}  \frac{\ol{M_0}}{z - z_0} +
O\left(\frac{1}{t^\alpha}\right).
\]
The Riemann--Hilbert problem for $m$ has jumps given by
\[
v(z) =\begin{cases}
M^C(z)^{-1}, & |z-z_0| = \eps,\\
M^C(z) \hat{v}(z) M^C(z)^{-1}, & z\in \hat{\Sigma}, \frac{\eps}{2} < |z-z_0| < \eps,\\
\id, & z\in \Sigma, |z-z_0|< \frac{\eps}{2},\\
\ti{M}^C(z)^{-1}, & |z^{-1}-z_0| = \eps,\\
\ti{M}^C(z) \hat{v}(z) \ti{M}^C(z)^{-1}, & z\in \hat{\Sigma}, \frac{\eps}{2} < |z^{-1}-z_0| < \eps,\\
\id, & z\in \Sigma, |z^{-1}-z_0|< \frac{\eps}{2},\\
\hat{v}(z), & \text{else}.
\end{cases}
\]
The jumps are $\id+O(t^{-1/2})$ on the loops $|z-z_0|=\eps$, $|z^{-1}-z_0|=\eps$ and
even $\id+O(t^{-\alpha})$ on the rest (in the $L^\infty$ norm, hence also in the $L^2$ one).
In particular, as in Lemma~\ref{lem:approrhp} we infer
\[
\|\mu - \rI \|_2 = O(t^{-1/2}).
\]
Thus we have with $\Omega_\infty$ as in \eqref{eq:defomegainfty}
\begin{align*}
m(z) = & \rI+ \frac{1}{2\pi\I} \int_{\Sigma} \mu(s) w(s) \Omega_\infty(s,z) \\
= & \rI + \frac{1}{2\pi\I} \int_{|s-z_0|=\eps} \mu(s) (M^C(s)^{-1}-\id) \Omega_\infty(s,z) \\
& + \frac{1}{2\pi\I} \int_{|s^{-1}-z_0|=\eps} \mu(s) (\ti{M}^C(s)^{-1}-\id) \Omega_\infty(s,z)
+ O(t^{-\alpha})\\
= & \rI + \frac{1}{(-2\sin(\theta_0) t)^{1/2}} \rI M_0 \frac{1}{2\pi\I} \int_{|s-z_0|=\eps} \frac{z_0}{s-z_0} \Omega_\infty(s,z)\\
& + \frac{1}{(-2\sin(\theta_0) t)^{1/2}} \rI \ol{M_0} \frac{1}{2\pi\I} \int_{|s^{-1}-z_0|=\eps} \frac{s}{s-\ol{z_0}} \Omega_\infty(s,z) + O(t^{-\alpha})\\
= & \rI - \frac{1}{(-2\sin(\theta_0) t)^{1/2}} ( m_0(z) + \bar{m}_0(z))
+ O\left(\frac{1}{t^\alpha}\right)
\end{align*}
finishing the proof.
\end{proof}

Hence, using \eqref{eq:AB} and \eqref{eq:checkm},
\be
\left( \check{m}(z) \right)_2 = \frac{1}{\ti{A}^2} \left( 1 + (T_1+2 B)z + O(z^2) \right)
\ee
and comparing with
\be
\left( \check{m}(z) \right)_2 = \left(1 - \frac{2\re(\beta)}{(-2\sin(\theta_0) t)^{1/2}}\right)
- \left(\frac{2\re(\ol{z_0}\beta)}{(-2\sin(\theta_0) t)^{1/2}}\right)z + O(z^2)
+O\left(\frac{1}{t^\alpha}\right),
\ee
we obtain
\be
\ti{A}^2 = 1 + \frac{2\re(\beta)}{(-2\sin(\theta_0) t)^{1/2}} + O\left(\frac{1}{t^\alpha}\right)
\ee
and
\be
T_1 + 2 B = -\frac{2\re(\ol{z_0}\beta)}{(-2\sin(\theta_0) t)^{1/2}}
+O\left(\frac{1}{t^\alpha}\right).
\ee
In summary we have
\begin{align}
A &= T_0 \left(1 + \frac{\re(\beta)}{(-2\sin(\theta_0) t)^{1/2}} + O\left(\frac{1}{t^\alpha}\right) \right),\\
B &= - \frac{1}{2} T_1 - \frac{\re(\ol{z_0}\beta)}{(-2\sin(\theta_0) t)^{1/2}}
+O\left(\frac{1}{t^\alpha}\right),
\end{align}
which proves Theorem~\ref{thm:asym2} in the analytic case.

\begin{remark}
Note that, in contradistinction to Theorem~\ref{thm:remcontour}, Theorem~\ref{thm:decoupling}
does not require uniform boundedness of the associated integral operators, but only some
knowledge of the solution of the Riemann-Hilbert problem. However, it requires that the
solution is of the form $\id+o(1)$ and hence cannot be used in the soliton region.
\end{remark}

\section{Analytic Approximation}
\label{sec:analapprox}

In this section we want to present the necessary changes in the case where the
reflection coefficient does not have an analytic extension. The idea is to
use an analytic approximation and split the reflection in an analytic part plus
a small rest. The analytic part will be moved to the complex plane while the rest
remains on the unit circle. This needs to be done in such a way that the rest
is of $O(t^{-l})$ and the growth of the analytic part can be controlled by the
decay of the phase.

In the soliton region a straightforward splitting based on the Fourier series
\be
R(z) = \sum_{k=-\infty}^{\infty} \hat{R}(k) z^k
\ee
will be sufficient. It is well-known that our assumption \eqref{decay} implies
$k^l \hat{R}(-k) \in \ell^1(\N)$ (this follows from the estimate \cite[eq.\ (10.83)]{tjac})
and $R \in C^l(\T)$.

\begin{lemma}\label{lem:analapprox}
Suppose $\hat{R}(k) \in \ell^1(\Z)$, $k^l \hat{R}(-k) \in \ell^1(\N)$ and let $0 < \eps < 1$, $\beta>0$ be given.
Then we can split the reflection coefficient according to
$R(z)= R_{a,t}(z) + R_{r,t}(z)$ such that $R_{a,t}(z)$ is analytic in $0<|z|<1$ and
\be
|R_{a,t}(z) \E^{-\beta t} | = O(t^{-l}), \quad 1-\eps \le |z|\le 1, \qquad
|R_{r,t}(z)| = O(t^{-l}), \quad |z|=1.
\ee
\end{lemma}

\begin{proof}
We choose $R_{a,t}(z) = \sum_{k = - K(t)}^\infty \hat{R}(k) z^k$ with $K(t) =
\lfloor \frac{\beta_0}{-\log(1-\eps)} t\rfloor$ for some positive $\beta_0<\beta$. Then, for $1-\eps \le|z|$,
\[
|R_{a,t}(z) \E^{-\beta t} | \le \sum_{k = - K(t)}^\infty |\hat{R}(k)| \E^{-\beta t} (1-\eps)^k \le
\|\hat{R}\|_1 \E^{-\beta t} (1-\eps)^{-K(t)} \le \|\hat{R}\|_1 \E^{-(\beta-\beta_0)t}
\]
Similarly, for $|z|=1$,
\[
|R_{r,t}(z) | \le \sum_{k = - \infty}^{-K(t)-1} |\hat{R}(k)| \le
const \sum_{k = K(t)+1}^\infty \frac{k^l}{K(t)^l}  |\hat{R}(-k)| \le \frac{const}{K(t)^l} \le
\frac{const}{t^l}.
\]
\end{proof}

To apply this lemma in the soliton region $z_0\in(-1,0)$ we choose
\be
\beta= \min_{|z|=1-\eps} -\re(\Phi(z))>0
\ee
and split $R(z) =  R_{a,t}(z) + R_{r,t}(z)$ according to Lemma~\ref{lem:analapprox} to obtain
\[
\ti{b}_\pm(z) = \ti{b}_{a,t,\pm}(z) \ti{b}_{r,t,\pm}(z) = \ti{b}_{r,t,\pm}(z) \ti{b}_{a,t,\pm}(z).
\]
Here $\ti{b}_{a,t,\pm}(z)$, $\ti{b}_{r,t,\pm}(z)$ denote the matrices obtained from $\ti{b}_\pm(z)$
as defined in \eqref{eq:deftib} by replacing $R(z)$ with $R_{a,t}(z)$, $R_{r,t}(z)$, respectively.
Now we can move the analytic parts into the complex plane as in Section~\ref{sec:conjdef}
while leaving the rest on $\T$. Hence, rather then \eqref{eq:jumpsolreg}, the jump now reads
\be
\hat{v}(z) = \begin{cases}
\ti{b}_{a,t,+}(z), & z\in\Sigma_+, \\
\ti{b}_{a,t,-}(z)^{-1}, & z\in\Sigma_-,\\
\ti{b}_{r,t,-}(z)^{-1} \ti{b}_{r,t,+}(z), & z\in\T.
\end{cases}
\ee
By construction we have $\hat{v}(z)= \id + O(t^{-l})$ on the whole contour and the rest follows as
in Section~\ref{sec:conjdef}.

In the other soliton region $z_0\in(0,1)$ we proceed similarly, with the only difference that the
jump matrices $\ti{B}_\pm(z)$ have at first sight more complicated off diagonal entries. To remedy
this we will rewrite them in terms of left rather then right scattering data. For this purpose
let us use the notation $R_r(z) \equiv R_+(z)$ for the right and $R_l(z) \equiv R_-(z)$ for the
left reflection coefficient. Moreover, let $T_r(z,z_0) \equiv T(z,z_0)$ be the right and
$T_l(z,z_0) \equiv T(z)/T(z,z_0)$ be the left partial transmission coefficient.

With this notation we have
\be
\ti{v}(z) = \begin{cases}
\ti{b}_-(z)^{-1} \ti{b}_+(z), \qquad \lam(z)> \lam_0,\\
\ti{B}_-(z)^{-1} \ti{B}_+(z), \qquad \lam(z)< \lam_0,\\
\end{cases}
\ee
where
\[
\ti{b}_-(z) = \begin{pmatrix} 1 & \frac{R_r(z^{-1}) \E^{-t\Phi(z)}}{T_r(z^{-1},z_0)^2} \\ 0 & 1 \end{pmatrix}, \quad
\ti{b}_+(z) = \begin{pmatrix} 1 & 0 \\ \frac{R_r(z) \E^{t\Phi(z)}}{T_r(z,z_0)^2}& 1 \end{pmatrix},
\]
and
\begin{align*}
\ti{B}_-(z) &= \begin{pmatrix} 1 & 0 \\ - \frac{T_{r,-}(z,z_0)^{-2}}{|T(z)|^2} R_r(z) \E^{t\Phi(z)} & 1 \end{pmatrix}, \\
\ti{B}_+(z) &= \begin{pmatrix} 1 & - \frac{T_{r,+}(z,z_0)^2}{|T(z)|^2} R_r(z^{-1}) \E^{-t\Phi(z)} \\ 0 & 1 \end{pmatrix}.
\end{align*}
Using \eqref{reltrpm} together with \eqref{eq:relB} we can further write
\begin{align*}
\ti{B}_-(z) &= \begin{pmatrix} 1 & 0 \\ \frac{R_l(z^{-1}) \E^{-t\Phi(z)}}{T_l(z^{-1},z_0)^2} & 1 \end{pmatrix}, \\
\ti{B}_+(z) &= \begin{pmatrix} 1 & \frac{R_l(z) \E^{t\Phi(z)}}{T_l(z,z_0)^2} \\ 0 & 1 \end{pmatrix}.
\end{align*}
Now we can proceed as before with $\ti{B}_\pm(z)$ as with $\ti{b}_\pm(z)$ by splitting $R_l(z)$ rather than $R_r(z)$.

In the similarity region we need to take the small vicinities of the stationary phase points into account. Since
the phase is quadratic near these points, we cannot use it to dominate the exponential growth of the analytic
part away from the unit circle. Hence we will take the phase as a new variable and use the Fourier transform
with respect to this new variable. Since this change of coordinates is singular near the stationary phase points,
there is a price we have to pay, namely, requiring additional smoothness for $R(z)$. We begin with

\begin{lemma}
Suppose $R(z)\in C^5(\T)$. Then we can split $R(z)$ according to
\be
R(z) = R_0(z) + (z-z_0)(z-\ol{z_0}) H(z), \qquad z \in \Sigma(z_0),
\ee
where $R_0(z)$ is a real polynomial in $z$ such that $H(z)$ vanishes at $z_0, \ol{z_0}$ of order three and
has a Fourier series
\be
H(z) = \sum_{k=-\infty}^\infty \hat{H}_k \E^{k \omega_0 \Phi(z)}, \qquad \omega_0=\frac{\pi}{\pi\cos(\theta_0)+\Phi_0},
\ee
with $k\,\hat{H}_k$ summable. Here $\Phi_0 = \Phi(z_0)/\I$.
\end{lemma}

\begin{proof}
By choosing a polynomial $R_0$ we can match the values of $R$ and its first four derivatives
at $z_0, \ol{z_0}$. Hence $H(z)\in C^4(\T)$ and vanishes together with its first three derivatives at
$z_0, \ol{z_0}$.

When restricted to $\Sigma(z_0)$ the phase $\Phi(z)/\I$ gives a one to one coordinate transform
$\Sigma(z_0) \to [\I\Phi_0,\I\Phi_0+\I\omega_0]$ and we can hence express $H(z)$ in this new coordinate.
The coordinate transform locally looks like a square root near $z_0$ and $\ol{z_0}$, however, due to our
assumption that $H$ vanishes there, $H$ is still $C^2$ in this new coordinate and the Fourier transform
with respect to this new coordinates exists and has the required properties.
\end{proof}

Moreover, as in Lemma~\ref{lem:analapprox} we obtain:

\begin{lemma}
Let $H(z)$ be as in the previous lemma. Then we can split $H(z)$ according to
$H(z)= H_{a,t}(z) + H_{r,t}(z)$ such that $H_{a,t}(z)$ is analytic in the region $\re(\Phi(z))<0$
and
\be
|H_{a,t}(z) \E^{\Phi(z) t/2} | = O(1), \: \re(\Phi(z))<0,|z|\le 1, \quad
|H_{r,t}(z)| = O(t^{-1}), \: |z|=1.
\ee
\end{lemma}

\begin{proof}
We choose $H_{a,t}(z) = \sum_{k = - K(t)}^\infty \hat{H}_k \E^{k \omega \Phi(z)}$ with $K(t) =
\lfloor t/(2\omega)\rfloor$. The rest follows as in Lemma~\ref{lem:analapprox}.
\end{proof}

By construction $R_{a,t}(z) = R_0(z) + (z-z_0)(z-\ol{z_0}) H_{a,t}(z)$ will satisfy the required
Lipschitz estimate in a vicinity of the stationary phase points (uniformly in $t$) and all
jumps will be $\id+O(t^{-1})$. Hence we can proceed as in Section~\ref{sec:reducecross}.

\appendix

\section{The solution on a small cross}
\label{sec:cross}

Introduce the cross $\Sigma = \Sigma_1 \cup\dots\cup \Sigma_4$ (see Figure~\ref{fig:contourcross}) by
\begin{align}
\nn \Sigma_1 & = \{u \E^{-\I\pi/4},\,u\in [0,\infty)\} &
\Sigma_2 & = \{u \E^{\I\pi/4},     \,u\in [0,\infty)\} \\
\Sigma_3 & = \{u \E^{3\I\pi/4},    \,u\in [0,\infty)\} &
\Sigma_4 & = \{u \E^{-3\I\pi/4},   \,u\in [0,\infty)\}.
\end{align}
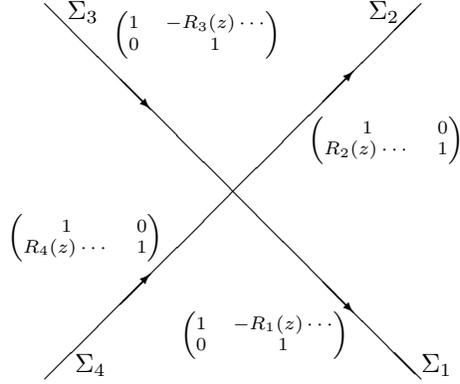
\begin{figure}
\begin{picture}(7,5.2)

\put(1,5){\line(1,-1){5}}
\put(2,4){\vector(1,-1){0.4}}
\put(4.7,1.3){\vector(1,-1){0.4}}
\put(1,0){\line(1,1){5}}
\put(2,1){\vector(1,1){0.4}}
\put(4.7,3.7){\vector(1,1){0.4}}

\put(6.0,0.1){$\Sigma_1$}
\put(5.3,4.8){$\Sigma_2$}
\put(1.3,4.8){$\Sigma_3$}
\put(1.4,0.1){$\Sigma_4$}

\put(2.8,0.5){$\scriptsize\begin{pmatrix} 1 & - R_1(z) \cdots\\ 0 & 1 \end{pmatrix}$}
\put(4.5,3.1){$\scriptsize\begin{pmatrix} 1 & 0 \\ R_2(z) \cdots & 1 \end{pmatrix}$}
\put(1.9,4.5){$\scriptsize\begin{pmatrix} 1 & - R_3(z) \cdots \\ 0 & 1 \end{pmatrix}$}
\put(0.5,1.8){$\scriptsize\begin{pmatrix} 1 & 0 \\ R_4(z) \cdots  & 1 \end{pmatrix}$}

\end{picture}
\caption{Contours of a cross}
\label{fig:contourcross}
\end{figure}%
Orient $\Sigma$ such that the real part of $z$ increases
in the positive direction. Denote by $\mathbb{D} = \{z,\,|z|<1\}$ the
open unit disc. Throughout this section $z^{\I\nu}$ will denote
the function $\E^{\I \nu \log(z)}$, where the branch cut of the logarithm is chosen along
the negative real axis $(-\infty,0)$.

Introduce the following jump matrices ($v_j$ for $z\in\Sigma_j$)
\begin{align}
\nn v_1 &= \begin{pmatrix} 1 & - R_1(z) z^{2\I\nu} \E^{- t \Phi(z)} \\ 0 & 1 \end{pmatrix}, &
v_2 &= \begin{pmatrix} 1 & 0 \\ R_2(z) z^{-2\I\nu} \E^{t \Phi(z)} & 1 \end{pmatrix},  \\
v_3 &= \begin{pmatrix} 1 & - R_3(z) z^{2\I\nu} \E^{- t \Phi(z)} \\ 0 & 1 \end{pmatrix}, &
v_4 &= \begin{pmatrix} 1 & 0 \\ R_4(z) z^{-2\I\nu} \E^{t \Phi(z)}  & 1 \end{pmatrix}.
\end{align}

Now consider the RHP given by
\begin{align}\label{eq:rhpcross}
m_+(z) &= m_-(z) v_j(z), && z\in\Sigma_j,\quad j=1,2,3,4,\\ \nn
m(z) &\to \id, && z\to \infty.
\end{align}
We have the next theorem, in which we follow the computations of
Sections 3 and 4 in \cite{dz}. The method can be found in earlier literature,
see for example \cite{its}.
One can also find arguments like this in Section 5 in \cite{km} or (3.65) to (3.76)
in \cite{dzp2}.

We will allow some variation, in all parameters as indicated in the next result.

\begin{theorem}\label{thm:solcross}
There is some $\rho_0>0$ such that $v_j(z)=\id$ for $|z|>\rho_0$. Moreover,
suppose that within $|z|\le\rho_0$ the following estimates hold:
\begin{enumerate}
\item
The phase satisfies $\Phi(0)=\I\Phi_0\in\I\R$, $\Phi'(0) = 0$, $\Phi''(0) = \I$ and
\begin{align}\label{estPhi}
\pm \re\big(\Phi(z)\big) &\geq \frac{1}{4} |z|^2,\quad
\begin{cases} + & \mbox{for } z\in\Sigma_1\cup\Sigma_3,\\ - &\mbox{else},\end{cases}\\ \label{estPhi2}
|\Phi(z) - \Phi(0) - \frac{\I z^2}{2}| &\leq C |z|^3.
\end{align}
\item
There is some $r\in\mathbb{D}$ and constants $(\alpha, L) \in (0,1] \times (0,\infty)$
such that $R_j$, $j=1,\dots,4$, satisfy H\"older conditions of the form
\begin{align}\nn
\abs{R_1(z) - \ol{r}} &\leq L |z|^\alpha, &
\abs{R_2(z) - r} &\leq L |z|^\alpha, \\\label{holdcondrj}
\abs{R_3(z) - \frac{\ol{r}}{1-\abs{r}^2}} &\leq L |z|^\alpha, &
\abs{R_4(z) - \frac{r}{1-\abs{r}^2}} &\leq L |z|^\alpha.
\end{align}
\end{enumerate}
Then the solution of the RHP \eqref{eq:rhpcross} satisfies
\be
m(z) = \id + \frac{1}{z} \frac{\I}{t^{1/2}} \begin{pmatrix} 0 & -\beta \\ \ol{\beta} & 0 \end{pmatrix}
+ O(t^{- \frac{1 + \alpha}{2}}),
\ee
for $|z|>\rho_0$, where
\be
\beta = \sqrt{\nu} \E^{\I(\pi/4-\arg(r)+\arg(\Gamma(\I\nu)))} \E^{-\I t \Phi_0} t^{-\I\nu},
\qquad \nu = - \frac{1}{2\pi} \log(1 - |r|^2).
\ee
Furthermore, if $R_j(z)$ and $\Phi(z)$ depend on some parameter, the error term is uniform
with respect to this parameter as long as $r$ remains within a compact subset of $\mathbb{D}$
and the constants in the above estimates can be chosen independent of the parameters.
\end{theorem}

We remark that the solution of the RHP \eqref{eq:rhpcross} is unique. This follows from the usual
Liouville argument \cite[Lem.~7.18]{deiftbook} since $\det(v_j)=1$.

Note that the actual value of $\rho_0$ is of no importance. In fact, if we choose $0 < \rho_1 < \rho_0$,
then the solution $\ti{m}$ of the problem with jump $\ti{v}$, where $\ti{v}$ is equal to $v$ for
$|z| < \rho_1$ and $\id$ otherwise, differs from $m$ only by an exponentially
small error.

This already indicates, that we should be able to replace $R_j(z)$ by their respective values at $z=0$.
To see this we start by rewriting our RHP as a singular integral equation.  We will use the theory
developed in Appendix~\ref{sec:sieq} for the case of $2\times2$ matrix valued functions with
$m_0(z)=\id$ and the usual Cauchy kernel (since we won't require symmetry in this section)
\[
\Omega(s,z) = \id \frac{ds}{s-z}.
\]
Moreover, since our contour is unbounded, we will assume $w\in L^1(\Sigma)\cap L^2(\Sigma)$.
All results from Appendix~\ref{sec:sieq} still hold in this case with some straightforward modifications
if one observes that $\mu-\id \in L^2(\Sigma)$. Indeed, as in Theorem~\ref{thm:cauchyop}, in the special case
$b_+(z) = v_j(z)$ and $b_-(z) = \id$ for $z\in\Sigma_j$, we obtain
\be\label{singintcross}
m(z) = \id + \frac{1}{2\pi\I} \int_{\Sigma} \mu (s) w(s) \frac{ds}{s - z},
\ee
where $\mu-\id$ is the solution of the singular integral equation
\be
(\id - C_w) (\mu -\id) = C_w \id,
\ee
that is,
\be\label{singinteqcross}
\mu = \id + (\id - C_w)^{-1} C_w \id, \qquad C_w f = \mathcal{C}_- (w f).
\ee
Here $\mathcal{C}$ denotes the usual Cauchy operator and we set $w(z)=w_+(z)$ (since $w_-(z)=0$).

As our first step we will get rid of some constants and rescale the
entire problem by setting
\be\label{scalecross}
\hat{m}(z) = D(t)^{-1} m(z t^{-1/2}) D(t),
\ee
where
\be
D(t) = \begin{pmatrix} d(t)^{-1} & 0 \\ 0 &d(t) \end{pmatrix}, \qquad
d(t) = \E^{\I t \Phi_0 /2} t^{\I\nu/2}, \quad d(t)^{-1} = \ol{d(t)}.
\ee
Then one easily checks that $\hat{m}(z)$ solves the RHP
\begin{align}
\hat{m}_+(z) &= \hat{m}_-(z) \hat{v}_j(z), && z\in\Sigma_j,\quad j=1,2,3,4,\\ \nn
\hat{m}(z) &\to \id, && z\to \infty,\quad z\notin\Sigma,
\end{align}
where $\hat{v}_j(z) = D(t)^{-1} v_j(z t^{-1/2}) D(t)$, $j=1,\dots,4$, explicitly
\begin{align} \nn
\hat{v}_1(z) &= \begin{pmatrix} 1 & -R_1(z t^{-1/2}) z^{2\I\nu} \E^{-t (\Phi(z t^{-1/2}) - \Phi(0))} \\ 0 & 1 \end{pmatrix},\\ \nn
\hat{v}_2(z) &= \begin{pmatrix} 1 & 0 \\ R_2(z t^{-1/2}) z^{-2\I\nu} \E^{t (\Phi(z t^{-1/2}) - \Phi(0))} & 1 \end{pmatrix}, \\ \nn
\hat{v}_3(z) &= \begin{pmatrix} 1 & -R_3(z t^{-1/2}) z^{2\I\nu} \E^{-t (\Phi(z t^{-1/2}) - \Phi(0))}\\ 0 & 1\end{pmatrix},\\
\hat{v}_4(z) &= \begin{pmatrix} 1 & 0 \\ R_2(z t^{-1/2}) z^{-2\I\nu} \E^{t (\Phi(z t^{-1/2}) - \Phi(0))} & 1 \end{pmatrix}.
\end{align}
Our next aim is to show that the solution $\hat{m}(z)$ of the rescaled problem is
close to the solution $\hat{m}^c(z)$ of the RHP
\begin{align}\label{eq:solrhpcross2}
\hat{m}^c_+(z) &= \hat{m}^c_-(z) \hat{v}^c_j(z), && z\in\Sigma_j,\quad j=1,2,3,4,\\ \nn
\hat{m}^c(z) &\to \id, && z\to \infty,\quad z\notin\Sigma,
\end{align}
associated with the following jump matrices
\begin{align} \nn
\hat{v}_1^c(z) &= \begin{pmatrix} 1 & -\ol{r} z^{2\I\nu} \E^{-\I z^2/2} \\ 0 & 1 \end{pmatrix}, &
\hat{v}_2^c(z) &= \begin{pmatrix} 1 & 0 \\ r z^{-2\I\nu} \E^{\I z^2/2} & 1 \end{pmatrix}, \\
\hat{v}_3^c(z) &= \begin{pmatrix} 1 & -\frac{\ol{r}}{1-\abs{r}^2} z^{2\I\nu} \E^{-\I z^2/2}\\ 0 & 1\end{pmatrix}, &
\hat{v}_4^c(z) &= \begin{pmatrix} 1 & 0 \\\frac{r}{1-\abs{r}^2} z^{-2\I\nu} \E^{\I z^2/2} & 1 \end{pmatrix}.
\end{align}
The difference between these jump matrices can be estimated as follows.

\begin{lemma}\label{lem:esticross}
The matrices $\hat{w}^c$ and $\hat{w}$ are close in the sense that
\be
\hat{w}_j(z) = \hat{w}^c_j(z) + O(t^{-\alpha/2}\E^{-|z|^2/8}),\quad
z \in\Sigma_j,\quad j=1,\dots 4.
\ee
Furthermore, the error term is uniform with respect to parameters as stated in Theorem~\ref{thm:solcross}.
\end{lemma}

\begin{proof}
We only give the proof $z\in\Sigma_1$, the other cases being similar. There is only
one nonzero matrix entry in $\hat{w}_j(z) - \hat{w}^c_j(z)$ given by
\[
W = \begin{cases} - R_1(z t^{-1/2}) z^{2\I\nu} \E^{-t(\Phi(z t^{-1/2}) - \Phi(0))} + \ol{r} z^{2\I\nu} \E^{-\I z^2/2},
& |z| \le \rho_0 t^{1/2},\\ \ol{r} z^{2\I \nu} \E^{-\I z^2/2} & |z| > \rho_0 t^{1/2}. \end{cases}
\]
A straightforward estimate for $|z| \le \rho_0 t^{1/2}$ shows
\begin{align*}
|W| & =  \E^{\nu\pi/4} |R_1(z t^{-1/2}) \E^{-t\hat{\Phi}(z t^{-1/2})} - \ol{r}| \E^{-|z|^2/2}\\
& \le  \E^{\nu\pi/4} |R_1(z t^{-1/2}) - \ol{r} | \E^{\re(-t\hat{\Phi}(z t^{-1/2}))-|z|^2/2} +
\E^{\nu\pi/4} |\E^{-t\hat{\Phi}(z t^{-1/2})}-1| \E^{-|z|^2/2}\\
& \le  \E^{\nu\pi/4} |R_1(z t^{-1/2}) - \ol{r} | \E^{-|z|^2/4} +
\E^{\nu\pi/4}  t |\hat{\Phi}(z t^{-1/2})| \E^{-|z|^2/4},
\end{align*}
where $\hat{\Phi}(z) = \Phi(z) - \Phi(0) - \frac{\I}{2} z^2 = \frac{\Phi'''(0)}{6} z^3 + \dots$.
Here we have used $\frac{\I}{2} z^2= \frac{1}{2}|z|^2$ for $z\in\Sigma_1$ and $\re(-t\hat{\Phi}(z t^{-1/2})) \le |z|^2/4$
by \eqref{estPhi}.
Furthermore, by \eqref{estPhi2} and \eqref{holdcondrj},
\begin{align*}
|W| & \le \E^{\nu\pi/4} L t^{-\alpha/2} |z|^\alpha \E^{-|z|^2/4} +
\E^{\nu\pi/4} C t^{-1/2} |z|^3 \E^{-|z|^2/4},
\end{align*}
for $|z| \le \rho_0 t^{1/2}$. For $|z| > \rho_0 t^{1/2}$ we have
\[
|W| \le \E^{\nu\pi/4} \E^{-|z|^2/2} \le \E^{\nu\pi/4} \E^{-\rho_0^2 t/4} \E^{-|z|^2/4}
\]
which finishes the proof.
\end{proof}

The next lemma allows us, to replace $\hat{m}(z)$ by $\hat{m}^c(z)$.

\begin{lemma}\label{lem:approrhp}
Consider the RHP
\begin{align}
m_+(z) &= m_-(z) v(z), && z\in\Sigma,\\ \nn
m(z) &\to \id, && z\to \infty,\quad z\notin\Sigma.
\end{align}
Assume that $w \in L^2(\Sigma) \cap L^{\infty}(\Sigma)$. Then
\be
\|\mu-\id\|_2 \le \frac{c\|w\|_2}{1-c\|w\|_\infty}
\ee
provided $c\|w\|_\infty<1$, where $c$ is the norm of the Cauchy operator on $L^2(\Sigma)$.
\end{lemma}

\begin{proof}
This follows since $\ti{\mu} = \mu-\id \in L^2(\Sigma)$ satisfies
$(\id - C_w) \ti{\mu} = C_w \id$.
\end{proof}

\begin{lemma}\label{lem:approcross}
The solution $\hat{m}(z)$ has a convergent asymptotic expansion
\be\label{asymhm}
\hat{m}(z) = \id + \frac{1}{z} \hat{M}(t) + O(\frac{1}{z^2})
\ee
for $|z|>\rho_0 t^{1/2}$ with the error term uniformly in $t$. Moreover,
\be
\hat{M}(t) = \hat{M}^c + O(t^{-\alpha/2}).
\ee
\end{lemma}

\begin{proof}
Consider $\hat{m}^d(z)= \hat{m}(z) \hat{m}^c(z)^{-1}$, whose jump matrix is given by
\[
\hat{v}^d(z) = \hat{m}_-^c(z)\hat{v}(z) \hat{v}^c(z)^{-1} \hat{m}_-^c(z)^{-1}=
\id + \hat{m}_-^c(z)\big(\hat{w}(z) - \hat{w}^c(z)\big)\hat{m}_-^c(z)^{-1}.
\]
By Lemma~\ref{lem:esticross}, we have that $\hat{w} - \hat{w}^c$ is
decaying of order $t^{-\alpha/2}$ in the norms of $L^1$ and $L^\infty$ and
thus the same is true for $\hat{w}^d = \hat{v}^d - \id$. Hence by the previous lemma
\[
\|\hat{\mu}^d - \id\|_2 = O(t^{-\alpha/2}).
\]
Furthermore, by $\hat{\mu}^d = \hat{m}^d_- = \hat{m}_- (\hat{m}^c_-)^{-1} =
\hat{\mu} (\hat{\mu}^c)^{-1}$ we infer
\[
\|\hat{\mu} - \hat{\mu}^c\|_2 = O(t^{-\alpha/2})
\]
since $ \hat{\mu}^c$ is bounded. Now
\[
\hat{m}(z) = \id - \frac{1}{2\pi\I} \frac{1}{z} \int_\Sigma \hat{\mu}(s) \hat{w}(s) ds +
\frac{1}{2\pi\I} \frac{1}{z} \int_\Sigma s \hat{\mu}(s) \hat{w}(s) \frac{ds}{s - z}
\]
shows (recall that $\hat{w}$ is supported inside $|z|\le\rho_0 t^{1/2}$)
\[
\hat{m}(z) = \id + \frac{1}{z} \hat{M}(t) + O(\frac{\|\hat{\mu}(s)\|_2 \|s \hat{w}(s)\|_2}{z^2}),
\]
where
\[
\hat{M}(t) = -\frac{1}{2\pi\I}  \int_\Sigma \hat{\mu}(s) \hat{w}(s) ds.
\]
Now the rest follows from
\[
\hat{M}(t) = \hat{M}^c - \frac{1}{2\pi\I}  \int_\Sigma (\hat{\mu}(s) \hat{w}(s) -\hat{\mu}^c(s) \hat{w}^c(s)) ds
\]
using $\|\hat{\mu} \hat{w} -\hat{\mu}^c \hat{w}^c\|_1 \le \|\hat{w}-\hat{w}^c\|_1 + \|\hat{\mu}-\id\|_2 \|\hat{w}-\hat{w}^c\|_2 +
\|\hat{\mu}-\hat{\mu}^c\|_2\|\hat{w}^c\|_2$.
\end{proof}

Finally, it remains to solve \eqref{eq:solrhpcross2} and to show:

\begin{theorem}
The solution of the RHP \eqref{eq:solrhpcross2}
is of the form
\be
\hat{m}^c(z) = \id + \frac{1}{z} \hat{M}^c + O(\frac{1}{z^2}),
\ee
where
\begin{align}
\hat{M}^c &= \I \begin{pmatrix} 0 & -\beta \\ \ol{\beta} & 0 \end{pmatrix},\quad
\beta = \sqrt{\nu} \E^{\I(\pi/4-\arg(r)+\arg(\Gamma(\I\nu)))}.
\end{align}
The error term is uniform with respect to $r$ in compact subsets of $\mathbb{D}$.
Moreover, the solution is bounded (again uniformly with respect to $r$).
\end{theorem}

Given this result, Theorem~\ref{thm:solcross} follows from Lemma~\ref{lem:approcross}
\begin{align}\nn
m(z) &= D(t) \hat{m}(z t^{1/2}) D(t)^{-1} = \id + \frac{1}{t^{1/2} z} D(t) \hat{M}(t) D(t)^{-1}+ O(z^{-2} t^{-1})\\
&= \id + \frac{1}{t^{1/2} z} D(t) \hat{M}^c D(t)^{-1}+ O(t^{-(1+\alpha)/2})
\end{align}
for $|z|>\rho_0$, since $D(t)$ is bounded.

The proof of this result will be given in the remainder of this section.
In order to solve \eqref{eq:solrhpcross2} we begin with a deformation which moves the jump to $\R$
as follows. Denote the region enclosed by $\mathbb{R}$ and $\Sigma_j$ as $\Omega_j$
(cf.\ Figure~\ref{fig:defcross})
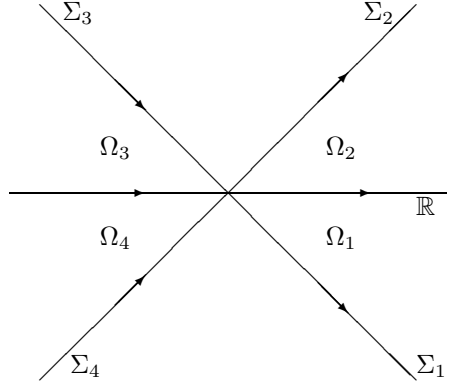
\begin{figure}
\begin{picture}(7,5.2)

\put(1,5){\line(1,-1){5}}
\put(2,4){\vector(1,-1){0.4}}
\put(4.7,1.3){\vector(1,-1){0.4}}
\put(1,0){\line(1,1){5}}
\put(2,1){\vector(1,1){0.4}}
\put(4.7,3.7){\vector(1,1){0.4}}
\put(0.6,2.5){\line(1,0){5.8}}
\put(2,2.5){\vector(1,0){0.4}}
\put(5,2.5){\vector(1,0){0.4}}

\put(6.0,0.1){$\Sigma_1$}
\put(5.3,4.8){$\Sigma_2$}
\put(1.3,4.8){$\Sigma_3$}
\put(1.4,0.1){$\Sigma_4$}
\put(6,2.2){$\R$}

\put(4.8,1.8){$\Omega_1$}
\put(4.8,3){$\Omega_2$}
\put(1.8,3){$\Omega_3$}
\put(1.8,1.8){$\Omega_4$}

\end{picture}
\caption{Deforming back the cross}
\label{fig:defcross}
\end{figure}%
and define
\be\label{eq:mrhpcross3}
\ti{m}^c(z) = \hat{m}^c(z) \begin{cases}
D_0(z) D_j, & z\in\Omega_j,\: j = 1,\dots,4,\\
D_0(z), & \mbox{else}, \end{cases}
\ee
where
\[
D_0(z) = \begin{pmatrix} z^{\I\nu} \E^{-\I z^2/4} & 0 \\ 0 & z^{-\I\nu} \E^{\I z^2/4} \end{pmatrix},
\]
and
\begin{align*}
D_1 &= \begin{pmatrix} 1 & \ol{r}  \\ 0 & 1 \end{pmatrix} &
D_2 &= \begin{pmatrix} 1 & 0 \\ r  & 1 \end{pmatrix} &
D_3 &= \begin{pmatrix} 1 & -\frac{\ol{r}}{1-\abs{r}^2} \\ 0 & 1\end{pmatrix} &
D_4 &= \begin{pmatrix} 1 & 0 \\ -\frac{r}{1-\abs{r}^2} & 1 \end{pmatrix}.
\end{align*}

\begin{lemma}\label{lem:reducetoline}
The function $\ti{m}^c(z)$ defined in \eqref{eq:mrhpcross3} satisfies the RHP
\begin{align}\label{eq:solrhpcross3}
\ti{m}^c_+ (z) &= \ti{m}^c_-(z) \begin{pmatrix} 1 - \abs{r}^2 & - \ol{r} \\ r & 1 \end{pmatrix},&&
z\in\mathbb{R}\\ \nn
\ti{m}^c(z) &= (\id + \frac{1}{z} \hat{M}^c + \dots) D_0(z),&& z\to\infty,
\: \frac{\pi}{4} < \arg(z) < \frac{3\pi}{4}.
\end{align}
\end{lemma}

\begin{proof}
First, one checks that $\ti{m}^c_+(z) = \ti{m}^c_-(z) D_0(z)^{-1} \hat{v}_1^c(z) D_0(z) D_1 =
\ti{m}^c_-(z)$, $z\in\Sigma_1$ and similarly for $z \in \Sigma_2, \Sigma_3, \Sigma_4$.
To compute the jump along $\R$ observe that, by our choice of branch cut for $z^{\I\nu}$, $D_0(z)$
has a jump along the negative real axis given by
\[
D_{0,\pm}(z) = \begin{pmatrix} \E^{(\log|z| \pm \I \pi) \I \nu} \E^{-\I z^2/4} & 0 \\ 0 & \E^{-(\log|z| \pm \I \pi) \I \nu} \E^{\I z^2/4} \end{pmatrix}, \qquad z<0.
\]
Hence the jump along $\R$ is given by
\[
D_1^{-1} D_2,\quad z>0 \quad\text{and}\quad
D_4^{-1} D_{0,-}^{-1}(z) D_{0,+}(z) D_3,\quad z<0,
\]
and \eqref{eq:solrhpcross3} follows after recalling $\E^{-2\pi\nu}=1-|r|^2$.
\end{proof}

Now, we can follow (4.17) to (4.51) in \cite{dz} to construct
an approximate solution.

The idea is as follows, since the jump matrix for \eqref{eq:solrhpcross3}, the derivative
$\frac{d}{dz} \ti{m}^c(z)$ has the same jump and hence is given by $n(z) \ti{m}^c(z)$,
where the entire matrix $n(z)$ can be determined from the behaviour $z\to\infty$. Since
this will just serve as a motivation for our ansatz, we will not worry about justifying any
steps.

For $z$ in the sector $\frac{\pi}{4} < \arg(z) < \frac{3\pi}{4}$ (enclosed by $\Sigma_2$ and $\Sigma_3$)
we have $\ti{m}^c(z) =\hat{m}^c(z) D_0(z)$ and hence
\begin{align*}
& \left( \frac{d}{dz} \ti{m}^c(z) + \frac{\I z}{2} \sigma_3 \ti{m}^c(z) \right) \ti{m}^c(z)^{-1}\\
& \qquad = \left( \I(\frac{\nu}{z} - \frac{z}{2}) \hat{m}^c(z) \sig_3 + \frac{d}{dz} \hat{m}^c(z)
+ \I \frac{z}{2} \sig_3 \hat{m}^c(z) \right) \hat{m}^c(z)^{-1}\\
& \qquad = \frac{\I}{2} [\sig_3,  \hat{M}^c] + O(\frac{1}{z}), \qquad
\sigma_3 = \begin{pmatrix} 1 & 0 \\ 0 & -1 \end{pmatrix}.
\end{align*}
Since the left hand side has no jump, it is entire and hence by Liouville's theorem a constant
given by the right hand side. In other words,
\be\label{eq:odetimc}
\frac{d}{dz} \ti{m}^c(z) + \frac{\I z}{2} \sigma_3 \ti{m}^c(z) = \beta \ti{m}^c(z),\quad
\beta = \begin{pmatrix} 0 & \beta_{12} \\ \beta_{21} & 0 \end{pmatrix} = \frac{\I}{2} [\sigma_3, \hat{M}^c].
\ee
This differential equation can be solved in terms of parabolic cylinder function which then gives
the solution of \eqref{eq:solrhpcross3}.

\begin{lemma}\label{lem:asymbyparabolic}
The RHP \eqref{eq:solrhpcross3} has a unique solution, and the term $\hat{M}^c$ is given by
\begin{align}
\hat{M}^c &= \I \begin{pmatrix} 0 & - \beta_{12} \\ \beta_{21} & 0 \end{pmatrix},\quad
&& \beta_{12} = \ol{\beta_{21}} = \sqrt{\nu} \E^{\I(\pi/4-\arg(r)+\arg(\Gamma(\I\nu)))}.
\end{align}
\end{lemma}

\begin{proof}
Uniqueness follows by the standard Liouville argument since the determinant of the jump matrix is equal to $1$.
To find the solution we use the ansatz
\[
\ti{m}^c(z) = \begin{pmatrix} \psi_{11}(z) & \psi_{12}(z) \\ \psi_{21}(z) & \psi_{22}(z) \end{pmatrix},
\]
where the functions $\psi_{jk}(z)$ satisfy
\begin{align*}
\psi_{11}''(z) &=  - \left(\frac{\I}{2} + \frac{1}{4}z^2-\beta_{12}\beta_{21}\right)\psi_{11}(z), &
\psi_{12}(z) &= \frac{1}{\beta_{21}}\left(\frac{d}{dz} - \frac{\I z}{2}\right)\psi_{22}(z),\\
\psi_{21}(z) &= \frac{1}{\beta_{12}}\left(\frac{d}{dz} + \frac{\I z}{2}\right)\psi_{11}(z), &
\psi_{22}''(z) &=  \left(\frac{\I}{2} - \frac{1}{4}z^2+\beta_{12}\beta_{21}\right)\psi_{22}(z).
\end{align*}
That is, $\psi_{11}(\E^{3\pi \I/4}\zeta)$ satisfies the parabolic cylinder equation
\[
D ''(\zeta) + \left(a + \frac{1}{2} - \frac{1}{4} \zeta^2\right) D(\zeta) = 0
\]
with $a= \I\beta_{12}\beta_{21}$ and $\psi_{22}(\E^{\I\pi/4}\zeta)$ satisfies the
parabolic cylinder equation with $a= -\I\beta_{12}\beta_{21}$.

Let $D_a$ be the entire parabolic cylinder function of \S16.5 in \cite{whwa} and set
\begin{align*}
\psi_{11}(z) &= \begin{cases}
\E^{-3\pi\nu/4}D_{\I\nu}(-\E^{\I\pi/4} z), & \im(z)>0,\\
\E^{\pi\nu/4} D_{\I\nu}(\E^{\I\pi/4} z), & \im(z)<0, \end{cases} \\
\psi_{22}(z) &= \begin{cases}
\E^{\pi\nu/4 } D_{-\I\nu}(-\I\E^{\I\pi/4} z), & \im(z)>0,\\
\E^{-3\pi\nu/4} D_{-\I\nu}(\I \E^{\I\pi/4} z), & \im(z)<0.\end{cases}
\end{align*}
Using the asymptotic behavior
\[
D_a(z) = z^a \E^{-z^2/4} \big(1 - \frac{a(a-1)}{2z^2} + O(z^{-4})\big),\quad z\to\infty,\quad \abs{\arg(z)} \leq 3 \pi/4,
\]
shows that the choice $\beta_{12}\beta_{21}=\nu$ ensures the correct asymptotics
\begin{align*}
&\psi_{11}(z) = z^{\I\nu} \E^{-\I z^2 /4} (1 + O(z^{-2})),
&&\psi_{12}(z) = -\I\beta_{12} z^{-\I\nu} \E^{\I z^2 /4} (z^{-1} + O(z^{-3})),\\
&\psi_{21}(z) = \I\beta_{21} z^{\I\nu} \E^{- \I z^2 /4} (z^{-1} + O(z^{-3})),
&&\psi_{22}(z) = z^{-\I\nu} \E^{\I z^2 /4} (1 + O(z^{-2})),
\end{align*}
as $z\to\infty$ inside the half plane $\im(z)\ge 0$. In particular,
\[
\ti{m}^c(z) = \big(\id + \frac{1}{z} \hat{M}^c + O(z^{-2})\big) D_0(z) \quad\mbox{with} \quad
\hat{M}^c = \I \begin{pmatrix} 0 & -\beta_{12} \\ \beta_{21} & 0 \end{pmatrix}.
\]
It remains to check that we have the correct jump. Since by construction both limits $\ti{m}^c_+(z)$ and
$\ti{m}^c_-(z)$ satisfy the same differential equation \eqref{eq:odetimc}, there is a constant matrix $v$
such that $\ti{m}^c_+(z) = \ti{m}^c_-(z) v$. Moreover, since the coefficient matrix of the linear differential
equation \eqref{eq:odetimc} has trace $0$, the determinant of $\ti{m}^c_\pm(z)$ is constant and hence
$\det(\ti{m}^c_\pm(z))=1$ by our asymptotics. Moreover, a straightforward calculation shows
\[
v= \ti{m}^c_-(0)^{-1} \ti{m}^c_+(0) = \begin{pmatrix}
\E^{-2\pi\nu} & - \frac{\sqrt{2\pi}\E^{-\I\pi/4} \E^{-\pi\nu/2}}{\sqrt{\nu}\Gamma(\I\nu)} \gamma^{-1}\\
\frac{\sqrt{2\pi}\E^{\I\pi/4} \E^{-\pi\nu/2}}{\sqrt{\nu}\Gamma(-\I\nu)} \gamma & 1
\end{pmatrix}
\]
where $\gamma =\frac{\sqrt{\nu}}{\beta_{12}}=\frac{\beta_{21}}{\sqrt{\nu}}$. Here we have used
\[
D_a(0)= \frac{2^{a/2}\sqrt{\pi}}{\Gamma((1-a)/2)}, \qquad
D_a'(0)= -\frac{2^{(1+a)/2}\sqrt{\pi}}{\Gamma(-a/2)}
\]
plus the duplication formula $\Gamma(z)\Gamma(z+\frac{1}{2})= 2^{1-2z} \sqrt{\pi} \Gamma(2z)$ for
the Gamma function. Hence, if we choose
\[
\gamma= \frac{\sqrt{\nu}\Gamma(-\I\nu)}{\sqrt{2\pi}\E^{\I\pi/4} \E^{-\pi\nu/2}} r,
\]
we have
\[
v= \begin{pmatrix} 1 - |r|^2 & - \ol{r} \\ r & 1 \end{pmatrix}
\]
since $|\gamma|^2=1$. To see this use $|\Gamma(-\I\nu)|^2 = \frac{\Gamma(1-\I\nu)\Gamma(\I\nu)}{-\I\nu}
= \frac{\pi}{\nu \sinh(\pi\nu)}$, which follows from Euler's reflection formula $\Gamma(1-z)\Gamma(z)=\frac{\pi}{\sin(\pi z)}$
for the Gamma function.

In particular,
\[
\beta_{12} = \ol{\beta_{21}}  = \sqrt{\nu} \E^{\I(\pi/4-\arg(r)+\arg(\Gamma(\I\nu)))},
\]
which finishes the proof.
\end{proof}

\begin{remark}
An inspection of the proof shows that $\hat{m}^c$ is given by the solution of a
differential equation depending analytically on $\nu$. Hence, $\hat{m}^c$
depends analytically on $\nu = - \frac{1}{2\pi} \log(1 - |r|^2)$. This
implies local Lipschitz dependence on $r$ as long as $r\in\mathbb{D}$.
\end{remark}

\section{Singular integral equations}
\label{sec:sieq}

In this section we show how to transform a meromorphic vector Riemann--Hilbert problem
with simple poles at $\zeta$, $\zeta^{-1}$,
\begin{align}\nn
& m_+(z) = m_-(z) v(z), \qquad z\in \Sigma,\\ \label{eq:rhp5m}
& \res_{\zeta} m(z) = \lim_{z\to\zeta} m(z)
\begin{pmatrix} 0 & 0\\ - \zeta \gam  & 0 \end{pmatrix},\quad
\res_{\zeta^{-1}} m(z) = \lim_{z\to\zeta^{-1}} m(z)
\begin{pmatrix} 0 & \zeta^{-1} \gam \\ 0 & 0 \end{pmatrix},\\ \nn
& m(z^{-1}) = m(z) \sigI,\\ \nn
& m(0) = \begin{pmatrix} 1 & m_2\end{pmatrix},
\end{align}
where $\zeta\in(-1,0)\cup(0,1)$ and $\gam\ge 0$, into a singular integral equation.
Since we require the symmetry condition for our Riemann--Hilbert
problems we need to adapt the usual Cauchy kernel to preserve this symmetry.
Moreover, we keep the single soliton as an inhomogeneous term which will play
the role of the leading asymptotics in our applications.

\begin{hypothesis}\label{hyp:sym}
Suppose the jump data $(\Sigma,v)$ satisfy the following assumptions:
\begin{enumerate}
\item
$\Sigma$ consist of a finite number of smooth oriented finite curves in $\C$
which intersect at most finitely many times with all intersections being transversal.
\item
$\Sigma$ does not contain $0$, $\zeta^{\pm 1}$.
\item
$\Sigma$ is invariant under $z\mapsto z^{-1}$ and is oriented such that under the
mapping $z\mapsto z^{-1}$ sequences converging from the positive sided to $\Sigma$
are mapped to sequences converging to the negative side.
\item
The jump matrix $v$ is invertible and can be factorized according to $v = b_-^{-1} b_+ =
(\id-w_-)^{-1}(\id+w_+)$, where $w_\pm = \pm(b_\pm-\id)$ are continuous and satisfy
\be\label{eq:wpmsym}
w_\pm(z^{-1}) = \sigI w_\mp(z) \sigI,\quad z\in\Sigma.
\ee
\end{enumerate}
\end{hypothesis}

The classical Cauchy-transform
of a function $f:\Sigma\to \C$ which is square integrable is the
analytic function $C f: \C\backslash\Sigma\to\C$ given by
\be
(C f)(z) = \frac{1}{2\pi\I} \int_{\Sigma} \frac{f(s)}{s - z} ds,\qquad z\in\C\backslash\Sigma.
\ee
Denote the non-tangential boundary values from both sides (taken possibly
in the $L^2$-sense --- see e.g.\ \cite[eq.\ (7.2)]{deiftbook}) by $C_+ f$ respectively $C_- f$.
Then it is well-known that $C_+$ and $C_-$ are bounded operators $L^2(\Sigma)\to L^2(\Sigma)$,
which satisfy $C_+ - C_- = \id$ and $C_+ C_- = 0$ (see e.g. \cite{bc}). Moreover, one has
the Plemelj--Sokhotsky formula (\cite{mu})
\[
C_\pm = \frac{1}{2} (\I H \pm \id),
\]
where
\be
(H f)(t) = \frac{1}{\pi} \dashint_\Sigma \frac{f(s)}{t-s} ds,\qquad t\in\Sigma,
\ee
is the Hilbert transform and $\dashint$ denotes the principal value integral.

In order to respect the symmetry condition we will restrict our attention to
the set $L^2_{s}(\Sigma)$ of square integrable functions $f:\Sigma\to\C^{2}$ such that
\be\label{eq:sym}
f(z^{-1}) = f(z) \sigI.
\ee
Clearly this will only be possible if we require our jump data to be symmetric as well (i.e.,
Hypothesis~\ref{hyp:sym} holds).

Next we introduce the Cauchy operator
\be
(C f)(z) = \frac{1}{2\pi\I} \int_\Sigma f(s) \Omega_\zeta(s,z)
\ee
acting on vector-valued functions $f:\Sigma\to\C^{2}$.
Here the Cauchy kernel is given by
\be
\Omega_{\zeta}(s,z) =
\begin{pmatrix} \frac{z-\zeta^{-1}}{s-\zeta^{-1}} \frac{1}{s-z} & 0 \\
0 & \frac{z-\zeta}{s-\zeta} \frac{1}{s-z} \end{pmatrix} ds =
\begin{pmatrix} \frac{1}{s-z} - \frac{1}{s-\zeta^{-1}} & 0 \\
0 & \frac{1}{s-z} - \frac{1}{s-\zeta} \end{pmatrix} ds,
\ee
for some fixed $\zeta\notin\Sigma$. In the case $\zeta=\infty$ we set
\be\label{eq:defomegainfty}
\Omega_{\infty}(s,z) =
\begin{pmatrix} \frac{1}{s-z} - \frac{1}{s} & 0 \\
0 & \frac{1}{s-z} \end{pmatrix} ds.
\ee
and one easily checks the symmetry property:
\be\label{eq:symC}
\Omega_\zeta(1/s,1/z) = \sigI \Omega_\zeta(s,z) \sigI.
\ee
The properties of $C$ are summarized in the next lemma.

\begin{lemma}
Assume Hypothesis~\ref{hyp:sym}.
The Cauchy operator $C$ has the properties, that the boundary values
$C_\pm$ are bounded operators $L^2_s(\Sigma) \to L^2_s(\Sigma)$
which satisfy
\be\label{eq:cpcm}
C_+ - C_- = \id
\ee
and
\be\label{eq:Cnorm}
(Cf)(\zeta^{-1}) = (0\quad\ast), \qquad (Cf)(\zeta) = (\ast\quad 0).
\ee
Here $\ast$ is a placeholder for an unspecified value.
Furthermore, $C$ restricts to $L^2_{s}(\Sigma)$, that is
\be
(C f) (z^{-1}) = (Cf)(z) \sigI,\quad z\in\C\backslash\Sigma
\ee
for $f\in L^2_{s}(\Sigma)$ and if $w_\pm$ satisfy \eqref{eq:wpmsym} we also have
\be \label{eq:symcpm}
C_\pm(f w_\mp)(1/z) = C_\mp(f w_\pm)(z) \sigI,\quad z\in\Sigma.
\ee
\end{lemma}

\begin{proof}
Everything follows from \eqref{eq:symC} and the fact that $C$ inherits all properties from
the classical Cauchy operator.
\end{proof}

We have thus obtained a Cauchy transform with the required properties.
Following Section 7 and 8 of \cite{bc} respectively \cite{krt}, we can solve our Riemann--Hilbert problem using this
Cauchy operator.

Introduce the operator $C_w: L_s^2(\Sigma)\to L_s^2(\Sigma)$ by
\be
C_w f = C_+(fw_-) + C_-(fw_+),\quad f\in L^2_s(\Sigma)
\ee
and recall from Lemma~\ref{lem:singlesoliton} that the unique
solution corresponding to $v\equiv \id$ is given by
\[
m_0(z)= \begin{pmatrix} f(z) & f(\frac{1}{z}) \end{pmatrix}, \quad
f(z) = \frac{1}{1 - \zeta^2 + \gamma}
\left(\gamma \zeta^2 \frac{z-\zeta^{-1}}{z - \zeta} + 1 - \zeta^2\right)
\]
Observe that for $\gam=0$ we have $f(z)=1$ and for $\gam=\infty$ we have
$f(z)= \zeta^2\frac{z-\zeta^{-1}}{z - \zeta}$. In particular, $m_0(z)$ is uniformly bounded
away from $\zeta$ for all $\gam\in[0,\infty]$.

Then we have the next result.

\begin{theorem}\label{thm:cauchyop}
Assume Hypothesis~\ref{hyp:sym}.

Suppose $m$ solves the Riemann--Hilbert problem \eqref{eq:rhp5m}. Then
\be\label{eq:mOm}
m(z) = (1-c_0) m_0(z) + \frac{1}{2\pi\I} \int_\Sigma \mu(s) (w_+(s) + w_-(s)) \Omega_\zeta(s,z),
\ee
where
\[
\mu = m_+ b_+^{-1} = m_- b_-^{-1} \quad\mbox{and}\quad
c_0= \left( \frac{1}{2\pi\I} \int_\Sigma \mu(s) (w_+(s) + w_-(s)) \Omega_\zeta(s,0) \right)_{\!1}.
\]
Here $(m)_j$ denotes the $j$'th component of a vector.
Furthermore, $\mu$ solves
\be\label{eq:sing4muc}
(\id - C_w) \mu = (1-c_0) m_0.
\ee

Conversely, suppose $\ti{\mu}$ solves
\be\label{eq:sing4mu}
(\id - C_w) \ti{\mu} = m_0,
\ee
and
\[
\ti{c}_0= \left( \frac{1}{2\pi\I} \int_\Sigma \ti{\mu}(s) (w_+(s) + w_-(s)) \Omega_\zeta(s,0) \right)_{\!1} \ne -1,
\]
then $m$ defined via \eqref{eq:mOm}, with $(1-c_0)=(1+\ti{c}_0)^{-1}$ and $\mu=(1+\ti{c}_0)^{-1}\ti{\mu}$,
solves the Riemann--Hilbert problem \eqref{eq:rhp5m} and $\mu= m_\pm b_\pm^{-1}$.
\end{theorem}

\begin{proof}
If $m$ solves \eqref{eq:rhp5m} and we set $\mu = m_\pm b_\pm^{-1}$,
then $m$ satisfies an additive jump given by
\[
m_+ - m_- = \mu (w_+ + w_-).
\]
Hence, if we denote the left hand side of \eqref{eq:mOm} by $\ti{m}$, both functions satisfy the same additive
jump. Furthermore, Hypothesis~\ref{hyp:sym} implies that $\mu$ is symmetric and hence so is $\ti{m}$. 
Using \eqref{eq:Cnorm} we also see that $\ti{m}$ satisfies the same pole conditions as $m_0$. In summary,
$m-\ti{m}$ has no jump and solves \eqref{eq:rhp5m} with $v\equiv \id$ except for the normalization which is given by 
$m(0)-\ti{m}(0)=\begin{pmatrix} 0 & *\end{pmatrix}$. Hence Lemma~\ref{lem:singlesoliton} implies $m-\ti{m} =0$.

Moreover, if $m$ is given by \eqref{eq:mOm}, then \eqref{eq:cpcm} implies
\begin{align} \label{eq:singtorhp}
m_\pm &= (1-c_0)m_0 + C_\pm(\mu w_-) + C_\pm(\mu w_+) \\
\nn &= (1-c_0)m_0 + C_w(\mu) \pm \mu w_\pm \\
\nn &= (1-c_0)m_0 - (\id - C_w) \mu + \mu b_\pm.
\end{align}
From this we conclude that $\mu = m_\pm b_\pm^{-1}$ solves \eqref{eq:sing4muc}.

Conversely, if $\ti{\mu}$ solves \eqref{eq:sing4mu}, then set
\[
\ti{m}(z) = m_0(z) + \frac{1}{2\pi\I} \int_\Sigma \ti{\mu}(s) (w_+(s) + w_-(s)) \Omega_\zeta(s,z),
\]
and the same calculation as in \eqref{eq:singtorhp} implies
$\ti{m}_\pm= \ti{\mu} b_\pm$, which shows that $m = (1+\ti{c}_0)^{-1} \ti{m}$ solves the
Riemann--Hilbert problem \eqref{eq:rhp5m}.
\end{proof}

Note that in the special case $\gamma=0$ we have $m_0(z)= \rI$ and
we can choose $\zeta$ as we please, say $\zeta=\infty$ such that $c_0=\ti{c}_0=0$
in the above theorem.

Hence we have a formula for the solution of our Riemann--Hilbert problem $m(z)$ in terms of
$(\id - C_w)^{-1} m_0$ and this clearly raises the question of bounded
invertibility of $\id - C_w$. This follows from Fredholm theory (cf.\ e.g. \cite{zh}):

\begin{lemma}
Assume Hypothesis~\ref{hyp:sym}.
The operator $\id-C_w$ is Fredholm of index zero,
\be
\ind(\id-C_w) =0.
\ee
\end{lemma}

\begin{proof}
Since one can easily check
\be
(\id-C_w) (\id-C_{-w}) = (\id-C_{-w}) (\id-C_w) = \id- T_w,
\ee
where
\[
T_w = T_{++} + T_{+-} + T_{-+} + T_{--}, \qquad
T_{\sig_1\sig_2}(f) = C_{\sig_1}[C_{\sig_2}(f w_{-\sig_2})w_{-\sig_1}] ,
\]
it suffices to check that the operators $T_{\sig_1\sig_2}$ are compact (\cite[Thm.~1.4.3]{proe}).
By Mergelyan's theorem we can approximate $w_\pm$ by rational functions and, since
the norm limit of compact operators is compact, we can assume without loss that $w_\pm$ have
an analytic extension to a neighborhood of $\Sigma$.

Indeed, suppose $f_n \in L^2(\Sigma)$ converges weakly to zero. Without loss we can assume
$f_n$ to be continuous. We will show that $\|T_w f_n\|_{L^2} \to 0$.

Using the analyticity of $w$ in a neighborhood of $\Sigma$
and the definition of $C_\pm$, we can slightly deform
the contour $\Sigma$ to some contour
$\Sigma_\pm$ close to $\Sigma$, on the left, and
have, by Cauchy's theorem,
\begin{align*}
T_{++} f_n(z) =& \frac{1}{2 \pi \I}
\int_{\Sigma_+} (C(f_n w_-)(s) w_-(s)) \Omega_\zeta(s,z).
\end{align*}
Now $(C(f_n w_-) w_-)(z) \to 0$ as $n \to \infty$.
Also
\[
|(C(f_n w_-) w_-)(z)| < const\, \|f_n\|_{L^2} \|w_- \|_{L^\infty}
< const
\]
and thus, by the dominated convergence theorem, $\|T_{++} f_n\|_{L^2} \to 0$ as
desired.

Moreover, considering $\id- \eps C_w = \id- C_{ \eps w}$ for $0 \leq \eps \leq 1$
we obtain $\ind(\id-C_w)= \ind(\id) =0$ from homotopy invariance of the index.
\end{proof}

By the Fredholm alternative, it follows that to show the bounded invertibility of $\id-C_w$
we only need to show that $\ker (\id-C_w) =0$. The latter being equivalent to
unique solvability of the corresponding vanishing Riemann--Hilbert problem in the case $\gam=0$
(where we can choose $\zeta=\infty$ such that $c_0=\ti{c}_0=0$).

\begin{corollary}
Assume Hypothesis~\ref{hyp:sym}.
A unique solution of the Riemann--Hilbert problem \eqref{eq:rhp5m} with $\gam=0$ exists if and only if
the corresponding vanishing Riemann--Hilbert problem, where the normalization condition is replaced by
$m(0)= \begin{pmatrix} 0 & m_2\end{pmatrix}$, with $m_2$ arbitrary, has at most one solution.
\end{corollary}

We are interested in comparing a Riemann--Hilbert problem for which $\|w\|_\infty$ is small with the one-soliton problem,
where
\be
\|w\|_\infty= \|w_+\|_{L^\infty(\Sigma)} + \|w_-\|_{L^\infty(\Sigma)}.
\ee
For such a situation we have the following result:

\begin{theorem}\label{thm:remcontour}
Fix a contour $\Sigma$ and choose $\zeta$, $\gam=\gam^t$, $v^t$ depending on some parameter $t\in\R$ such that
Hypothesis~\ref{hyp:sym} holds.

Assume that  $w^t$ satisfies
\be
\|w^t\|_\infty \leq \rho(t)
\ee
for some function $\rho(t) \to 0$ as $t\to\infty$. Then $(\id-C_{w^t})^{-1}: L^2_s(\Sigma)\to L^2_s(\Sigma)$ exists
for sufficiently large $t$ and the solution $m(z)$ of the Riemann--Hilbert problems \eqref{eq:rhp5m} differs
from the one-soliton solution $m_0^t(z)$ only by $O(\rho(t))$, where the error term depends on the distance of $z$
to $\Sigma \cup \{\zeta^{\pm 1}\}$.
\end{theorem}

\begin{proof}
By boundedness of the Cauchy transform, one has
\[
\|C_{w^t}\| \leq const \|w^t\|_\infty.
\]
Thus, by the Neumann series, we infer that $(\id-C_{w^t})^{-1}$ exists for sufficiently large $t$ and
\[
\|(\id-C_{w^t})^{-1}-\id\| = O(\rho(t)).
\]
This implies $\|\ti{\mu}^t - m_0^t\|_{L^2_s} = O(\rho(t))$ and $\ti{c}_0^t = O(\rho(t))$ (note $\ti{\mu}_0^t = \mu_0^t = m_0^t$).
Consequently $c_0^t = O(\rho(t))$ and $\|\mu^t - m_0^t\|_{L^2_s} = O(\rho(t))$ and thus $m^t(z) - m_0^t(z) = O(\rho(t))$
uniformly in $z$ as long as it stays a positive distance away from  $\Sigma \cup \{\zeta^{\pm 1}\}$.
\end{proof}

\noindent
{\bf Acknowledgments.}
We thank Ira Egorova, Katrin Grunert, Alice Mikikits-Leitner, and Johanna Michor
for pointing out errors in a previous version of this article. Furthermore, we are indebted to Fritz Gesztesy and
the anonymous referee for valuable suggestions improving the presentation of the material.

\end{document}